\documentclass[nospthms]{article}

\usepackage[authoryear]{natbib}

\usepackage{graphicx}
\usepackage{amsthm, latexsym, amsfonts, amsmath}
\usepackage{times}
\usepackage{multirow}

\usepackage{verbatim}
\usepackage{subfigure}
\usepackage{ragged2e}

\usepackage{xr}
\newtheorem{theorem}{Theorem}
\newtheorem{lemma}{Lemma}
\newtheorem{corollary}{Corollary}
\newtheorem{proposition}{Proposition}

\theoremstyle{definition}
\newtheorem{definition}{Definition}

\newtheorem*{theorem*}{Theorem}
\newtheorem*{lemma*}{Lemma}

\begin{document}\sloppy
\title{Connectomic Constraints on Computation in \\ Feedforward Networks of Spiking Neurons}

\author{Venkatakrishnan Ramaswamy\thanks{Computer and Information Science and Engineering, University of Florida, Gainesville, FL 32611, USA. {\em Present Address:} Interdisciplinary Center for Neural Computation, The Hebrew University of Jerusalem, Jerusalem 91904, Israel. {\tt venkat.ramaswamy@mail.huji.ac.il}} \and Arunava Banerjee\thanks{Computer and Information Science and Engineering, University of Florida, Gainesville, FL 32611, USA. {\tt arunava@cise.ufl.edu}}}

\maketitle

\begin{abstract}
Several efforts are currently underway to decipher the connectome or parts thereof in a variety of organisms. Ascertaining the detailed physiological properties of all the neurons in these connectomes, however, is out of the scope of such projects. It is therefore unclear to what extent knowledge of the connectome alone will advance a mechanistic understanding of computation occurring in these neural circuits, especially when the high-level function of the said circuit is unknown. We consider, here, the question of how the wiring diagram of neurons imposes constraints on what neural circuits can compute, when we cannot assume detailed information on the physiological response properties of the neurons. We call such constraints -- that arise by virtue of the connectome -- {\em connectomic constraints} on computation. For feedforward networks equipped with neurons that obey a deterministic spiking neuron model which satisfies a small number of properties, we ask if just by knowing the architecture of a network, we can rule out computations that it could be doing, no matter what response properties each of its neurons may have. We show results of this form, for certain classes of network architectures. On the other hand, we also prove that with the limited set of properties assumed for our model neurons, there are fundamental limits to the constraints imposed by network structure. Thus, our theory suggests that while connectomic constraints might restrict the computational ability of certain classes of network architectures, we may require more elaborate information on the properties of neurons in the network, before we can discern such results for other classes of networks. 
%\keywords{Spiking neurons \and Connectomics \and Feedforward networks}
\end{abstract}

%%%%%%%%%%%

\section{Introduction}
Recent remarkable experimental advances \citep{denkhorst,hayworth,knott,mishchenko2010ultrastructural,turaga2010convolutional,helmstaedter2011high,mikula2012staining} have brought the prospect of ascertaining the connectome or parts thereof closer to reality \citep{chklovskii2010semi,kleinfeld2011large,seung2011neuroscience,denk,reid2012functional,helmstaedter2013connectomic}. This data is currently not expected to include information on
the detailed physiological properties of all the neurons in the
connectome. Even so, already, there have been two pioneering studies \citep{briggman2011wiring,bock2011network} that fruitfully use electron-microscopy reconstructions in conjunction with two-photon calcium imaging on the same tissue. In \citep{briggman2011wiring}, the authors used this approach to rule out certain models of direction
selectivity in the retina. The other study \citep{bock2011network} examined the orientation-selectivity circuitry in the cortex and found that inhibitory interneurons received convergent anatomical input from nearby excitatory neurons that had a broad range of preferred orientations. Recent work \citep{takemura2013visual} has also used connectomic reconstructions of the motion detection circuit in the fruit fly visual system, in order to identify cellular targets for future functional investigations; this is towards the goal of a comprehensive mechanistic understanding of this circuit. While this broad approach of combining functional imaging with structural reconstructions creates new opportunities to unravel structure-function relationships \citep{seung2011neuroscience}, to fruitfully use functional imaging seems to require that (a) we have an a priori credible hypothesis about at least one high-level computation that the neural circuit in question is performing and (b) we have a way of experimentally eliciting performance of the said computation, usually via an appropriate stimulus. Unfortunately, neither of these conditions appear to be satisfied for a majority of neuronal circuits in the brain, especially as one moves away from the sensory/motor periphery. Suppose, in addition to its wiring diagram, we knew the detailed physiological response properties of all the neurons in such a neural circuit to the extent that we could predict circuit behavior (via simulations, for example). This might provide a way forward towards advancing hypotheses about what high-level computation(s) the circuit is actually involved in. Regrettably, ascertaining the detailed physiological response properties of all the neurons in such a network appears to be out of reach of current experimental technology. The prospects of obtaining the wiring diagram, however, seem to hold more promise. The question therefore becomes: (1) What can we learn from the wiring diagram alone, even when the specific high-level function of the neural circuit may be unknown? (2) Are there fundamental limits to what can be
learned from the wiring diagram alone, in the absence of more detailed
physiological information?

To investigate these questions, we have studied a network model
equipped with neurons that obey a deterministic spiking neuron
model. We ask what computations networks of specific architectures
{\em cannot} perform, no matter what response properties each of their
neurons may have. The implication, then, is that, owing to its
structure, the network is unable to effect the computation in
question. That is, connectomic constraints forbid the network from
performing the said computation. In addition, to rule out the possibility
that this computation is so ``hard'' that no network (of
any architecture) can accomplish it, we stipulate the need to demonstrate that
there exists a network (of a different architecture) comprising simple
neurons that can indeed effect this computation. The goal of this
paper is to establish results of this form for various network
architectures, after setting up a mathematical framework within which
these questions can be precisely posed. As a first simplifying step,
in this paper, we limit our study to feedforward networks of neurons. Having started with this goal, however, we also find that with the small number of basic properties assumed for our model neurons, there are fundamental limits to the computational constraints imposed by network structure, in certain cases. In
particular, we prove that, constrained only by the properties in the
current neuron model, every feedforward network, of arbitrary size and depth, has an
equivalent feedforward network of depth equal to two that effects {\em
  exactly} the same computation. The implication of this result is
that we need more elaborate information about the properties of the
neurons before connectomic constraints on the computational ability of
such networks can be discerned.

Before we can examine these questions, we are confronted with the problem of having to define what
computation exactly means, in this context. Physically, neurons and their
networks are simply devices that receive spike-trains as input, and in
turn generate spike-trains as output. It is this translation from
spike-trains to spike-trains that characterizes information processing
and indeed even cognition in the brain. It is tempting to view a
feedforward network as a {\em transformation}, which is to say a
function, that associates a {\em unique} output spike train with each combination of afferent input spike trains, since such networks do not have recurrent loops. This is
the intuition we will seek to make precise.

Since the functional role of single neurons and small networks in the
brain is not yet well understood, we do not make assumptions about
particular high-level tasks that the network is trying to perform; we are
just interested in physical spike-train to spike-train
transformations. Likewise, since the kinds of neural code employed are
unclear, we make no overarching assumptions about the neural code
either. We study precise spike times since there is widespread
evidence \citep[\& references therein]{streh, spikes} that precise spike times play a role in
information processing in the brain, in many cases. Indeed, Spike-Timing Dependent Plasticity, a class of Hebbian learning rules that are sensitive to the relative timing of pre and postsynaptic spikes have been discovered \citep{markram1997regulation,bi1998synaptic} that support the role of precise spike-timing in computation in the brain. Studying spike
times also subsumes cases where spiking rate may be the relevant
parameter and therefore there is no loss of generality in making this
assumption.

\section{Notation and Preliminaries}\label{sec:notation}
In this section, we define the mathematical formalism used to
describe spike-trains and frequently-used operations on them that, for instance, shift
and segment them. The reader may skim these on the first reading and
revisit them if a specific technical point needs clarification later
on.

An {\em action potential} or {\em spike} is a stereotypical event
characterized by the time instant at which it is
initiated in the neuron, which is referred to as its {\em spike
  time}. Spike times are represented relative to the present by real
numbers, with positive values denoting past spike times and negative
values denoting future spike times. A {\em spike-train} ${\vec x}=\langle x^1, x^2, \ldots, x^k,
\ldots\rangle$ is a strictly increasing sequence of spike times, with every
pair of spike times being at least $\alpha$ apart, where $\alpha>0$
is the absolute refractory period\footnote{We assume a single
  fixed absolute refractory period for all neurons, for convenience,
  although our results would be no different if different neurons had
  different absolute refractory periods.} and $x^i$ is the spike time
of spike $i$. An {\em empty spike-train}, denoted by ${\vec \phi}$, is
one which has no spikes. A {\em time-bounded spike-train} (with {\em bound} $(a,b)$) is
one where all spike times lie in the bounded interval $(a,b)$, for
some $a,b \in \mathbb{R}$. We use ${\cal S}$ to denote the set of all spike
trains and $\bar{\cal S}_{(a, b)}$ to denote the set of all
time-bounded spike-trains with bound $(a, b)$. A spike-train is said
to have a {\em gap} in the interval $(c,d)$, if it has no spikes in
that time interval. Furthermore, this gap is said to be of {\em
  length} $d-c\/$. 

We use the term {\em spike-train ensemble} to denote a collection of
spike-trains. Thus, formally, a {\em spike-train ensemble}
$\chi=\langle{\vec x_1}, \ldots, {\vec x_m}\rangle$ is a tuple of
spike-trains. The {\em order} of a spike-train ensemble is the number
of spike-trains in it. For example, $\chi=\langle{\vec x_1}, \ldots,
{\vec x_m}\rangle$ is a spike-train ensemble of order $m$. A {\em
  time-bounded spike-train ensemble} (with {\em bound} $(a,b)$) is one
in which each of its spike-trains is time-bounded (with {\em bound}
$(a,b)$). A
spike-train ensemble $\chi$ is said have a {\em gap} in the interval
$(c,d)$, if each of its spike trains has a gap in the interval
$(c,d)$.

Next, we define some operators to time-shift, segment and
assemble/disassemble spike-trains from spike-train ensembles. Let
${\vec x}=\langle x^1, x^2, \ldots, x^k, \ldots\rangle$ be a
spike-train and $\chi=\langle{\vec x_1}, \ldots, {\vec x_m}\rangle$ be
a spike-train ensemble.  The {\em time-shift operator for
  spike-trains} is  used to time-shift all the spikes in a spike-train. Thus, $\sigma_t({\vec x}) =\langle x^1-t,
x^2-t, \ldots, x^k-t, \ldots\rangle$. The {\em time-shift operator for
  spike-train ensembles} is defined as $\sigma_t(\chi)
=\langle\sigma_t({\vec x_1}), \ldots, \sigma_t({\vec
  x_m})\rangle$. The {\em truncation operator for spike-trains} is
used to ``cut out'' specific segments of a spike-train. It is defined as follows: $\Xi_{[a,b]}({\vec x})$ is the time-bounded
spike-train with bound $[a,b]$ that is identical to ${\vec x}$ in the
interval $[a,b]$. $\Xi_{(a,b)}({\vec x})$, $\Xi_{(a,b]}({\vec x})$ and
  $\Xi_{[a,b)}({\vec x})$ are defined likewise. In the same vein,
    $\Xi_{[a,\infty)}({\vec x})$ is the spike-train that is identical
      to ${\vec x}$ in the interval $[a,\infty)$ and has no spikes in
        the interval $(-\infty, a)$. Similarly, $\Xi_{(-\infty,
          b]}({\vec x})$ is the spike-train that is identical to
      ${\vec x}$ in the interval $(-\infty, b]$ and has no spikes in
    the interval $(b, \infty)$. $\Xi_{(a,\infty)}({\vec x})$ and
    $\Xi_{(-\infty, b)}({\vec x})$ are also defined similarly. The
    {\em truncation operator for spike-train ensembles} is defined as
    $\Xi_{[a,b]}(\chi)=\langle{\Xi_{[a,b]}(\vec x_1}), \ldots,
    \Xi_{[a,b]}({\vec x_m})\rangle$. $\Xi_{(a,b)}(\chi)$,
    $\Xi_{(a,b]}(\chi)$, $\Xi_{[a,b)}(\chi)$,
    $\Xi_{[a,\infty)}(\chi)$, $\Xi_{(-\infty, b]}(\chi)$,
    $\Xi_{(a,\infty)}(\chi)$ and $\Xi_{(-\infty, b)}(\chi)$ are
    defined likewise. Furthermore, $\Xi_t(\cdot)$ is shorthand for
    $\Xi_{[t,t]}(\cdot)$. The {\em projection operator for spike-train
      ensembles} is used to ``pull-out'' a specific spike-train from a
    spike-train ensemble. It is defined as $\Pi_i(\chi) = {\vec x_i}$,
    where $1\leq i \leq m$. Let ${\vec y_1}, {\vec y_2}, \ldots, {\vec
      y_n}$ be spike-trains. The {\em join operator for spike-trains}
    is used to ``bundle-up'' a set of spike-trains to obtain a spike-train ensemble. It is defined as ${\vec y_1}\sqcup {\vec y_2}\sqcup \ldots\sqcup
    {\vec y_n} = \mathop {\bigsqcup} \limits_{i=1}^{n}{\vec y_i} =
    \langle{\vec y_1}, {\vec y_2}, \ldots, {\vec y_n} \rangle$.

\section{The Neuron Model}\label{sec:Model}
The present work treats the setting in which we know the wiring diagram of a network, but lack detailed information on the response properties of its neurons. We then wish to show computations that the network cannot accomplish, {\em no matter what response properties its neurons may have}. The modeling question we must first address, therefore, is what kind of neuron model we ought to use in such a context.

While we lack detailed information on each of the neurons in the network, it is reasonable to assume that all the neurons in the network satisfy a small number of elementary properties. For example, spiking neurons are generally known to have an absolute refractory period and most of them settle to a resting membrane potential upon receiving no input for sufficiently long, where this resting membrane potential is smaller than the threshold required to elicit a spike. We wish to have a model that is contingent on a small number of such basic properties, but whose responses are unconstrained otherwise, in order to allow for a large class of possible responses. 

Mathematically, we formulate the neuron as an abstract mathematical
object that satisfies a small number of axioms, which correspond to
such elementary properties.

Another way to think about the model is as one that brings ``under its umbrella'' several other neuron models. These are models that satisfy the properties  that our model is contingent on. In Appendix A, we demonstrate, for instance, that neuron models such as the Leaky Integrate-and-Fire Model and the Spike Response Model SRM$_0$ satisfy these properties up to arbitrary accuracy. Our model can thus be seen as a generalization\footnote{Models such as the Leaky Integrate-and-Fire (LIF) and Spike Response Model (SRM), in addition to the constraints in our model have their membrane potential function $P(\cdot)$ specified outright. In case of the LIF model, this is specified via a differential equation and in the case of SRM, the specific functional form is written down explicitly.} of these neuron models, specifically one that allows for a much wider class of responses. 

There are also other strong reasons for employing this type of model. Crucially, it allows the possibility of incrementally adding more properties to the neuron model, and studying how that further constrains the computational properties of the network. This would model the scenario where we have more detailed knowledge about individual neuron properties, which might well turn out to be the case with the connectome projects. While technical hurdles presently lie in the way of inferring, for example, distributions of ion-channels and neurotransmitter receptors in each neuron using electron microscopy\citep{denk}, it is conceivable that future advances make this possible, giving us a better sense of the physiological properties of all the individual neurons in the connectome; other future technological advances may also help in this direction. Furthermore, the need for adding more properties to the model and studying the consequences will become especially apparent towards the end of this paper, when we show limits to the constraints imposed by the present set of properties assumed in the model.

\subsection{Properties}
We start off by informally describing the properties that our model is
contingent on. Notable cases where the properties do not hold are also
pointed out. This is followed by a formal mathematical definition of the model. The approach
taken here in defining the model is along the lines of the one in
\citep{arun2}.

The following are our assumptions:

\begin{enumerate}
\item We assume that the neuron is a device that receives input from
  other neurons exclusively by spikes which are received via chemical
  synapses.\footnote{In this work, we do not treat electrical
    synapses or ephaptic interactions \citep{shep}.} 

\item The neuron is a finite-precision device with fading
  memory. Hence, the underlying potential function can be determined\footnote{We do not treat stochastic variability in the responses of neurons or neuromodulation in this paper.}
  from a bounded past. That is, we assume that, for each neuron, there
  exist positive real numbers $\Upsilon$ and $\rho$, so that the current
  membrane potential of the neuron can be determined as a function of
  the input spikes received in the past $\Upsilon$ milliseconds and the
  spikes produced by the neuron in the past $\rho$ milliseconds. The parameter $\Upsilon$ would correspond to the timescale at which the neuron integrates inputs received from other neurons and $\rho$
  corresponds to the notion of {\em relative refractory period.}

\item Specifically, we assume that the membrane potential of the neuron can be
  written down as a real-valued, everywhere-bounded
  function of the form $P(\chi; {\vec x_0})$, where ${\vec x_0}$ is a time-bounded
  spike-train, with bound $(0, \rho)$ and $\chi=\langle{\vec x_1}, \ldots,
  {\vec x_m}\rangle$ is a time-bounded spike-train ensemble with
  bound $(0, \Upsilon)$. Informally, ${\vec x_i}$, for $1\leq i \leq m$, is the
  sequence of spikes afferent in synapse $i$ in the past $\Upsilon$
  milliseconds and ${\vec x_0}$ is the sequence of spikes efferent from the
  current neuron in the past $\rho$ milliseconds.
  The function $P(\cdot)$ characterizes the entire
  spatiotemporal response of the neuron to spikes including synaptic
  strengths, their location on dendrites, and
  their modulation of each other's effects at the soma,
  spike-propagation delays, and the postspike hyperpolarization.

\item Without loss of generality, we assume the
  resting membrane potential to be
  $0$. 

\item Let $\tau > 0$ be the threshold that the membrane potential must
  reach in order to elicit a spike. Observe that the model allows for
  variable\footnote{In many biological neurons, the membrane potential that the soma (or axon initial segment) must reach, in order to elicit a spike is not fixed at all times and is, for example, a function of the inactivation levels of the voltage-gated Sodium channels. Our model can accomodate this phenomenon, to the extent that this threshold itself is a function of spikes afferent in the past $\Upsilon$ milliseconds and spikes efferent from the present neuron in the past $\rho$ milliseconds.} thresholds, as long as the threshold itself is a function
  of spikes afferent in the past $\Upsilon$ milliseconds and spikes
  efferent from the present neuron in the past $\rho$
  milliseconds. Furthermore, when a new output spike is produced, in
  the model, the membrane potential immediately goes below
  threshold. That is, the membrane potential function in the model
  takes values that are at most that of the threshold. This simplifies
  our condition for an output spike to be that the $P(\cdot)$ merely
  hits threshold, without having to check if it hits it from below,
  since it cannot hit it from above. Again, this is done without loss of generality. Additionally, let $\lambda$ be a negative real number that represents a lower-bound on the values that the membrane potential can take.

\item Output spikes in the recent past tend to have an inhibitory effect, in the following
  sense\footnote{This is violated, notably, in neurons that have
    a post-inhibitory rebound.}: \\$P(\chi; {\vec x_0})
  \leq P(\chi ; {\vec \phi})$, for all ``legal'' $\chi$ and ${\vec x_0}$.

Thus, our model allows for a wide variety of AHPs. Indeed, the only constraint on AHPs is the one given above. That is, suppose, in the first case that at a certain point in time the neuron received spikes in the past $\Upsilon$ seconds present in $\chi$ as input and did not output any spikes in the past $\rho$ milliseconds. In the second case, suppose that at a certain point in time the neuron again received spikes in the past $\Upsilon$ seconds present in $\chi$ as input but output some spikes in the past $\rho$ milliseconds. The condition states that the membrane potential in the second case must be at most that of the value in the first case. Thus, our results will be true for any neuron model that has an AHP that obeys this condition.

\item Owing to the absolute refractory period $\alpha>0$, no two input
  or output spikes can occur closer than $\alpha$. That is, suppose ${\vec x_0}=\langle x_0^1, x_0^2,
  \ldots, x_0^k \rangle$, where $x_0^1 < \alpha$. Then $P(\chi; {\vec x_0})<\tau$, for all ``legal'' $\chi$.

\item Finally, on receiving no input spikes in the past $\Upsilon$
  milliseconds and no output spikes in the past $\rho$ milliseconds, the neuron settles to its resting potential. That is,\\
$P(\langle{\vec \phi}, {\vec \phi}, \ldots, {\vec \phi} \rangle; {\vec \phi}) = 0$.\\

\end{enumerate}
A {\em feedforward network of neurons}, is a Directed Acyclic Graph
where each vertex corresponds to an instantiation of the neuron model, with the exception of some vertices, designated as input vertices (which
are placeholders for input spike-trains); one neuron is
designated the output neuron. The {\em order} of a feedforward network is equal to the number of its input vertices. The {\em depth} of a feedforward network is
the length of the longest path from an input vertex to the output
vertex.\\

\noindent
Next, we formalize the above notions into a rigorous definition of a neuron
as an abstract mathematical object. 

\begin{definition}[Neuron]
A \emph{neuron} ${\mathsf N}$ is a 7-tuple $\langle\alpha, \Upsilon, \rho, \tau, \lambda, m, 
P:{\bar{\cal S}_{(0, \Upsilon)}}^m \times \bar{\cal S}_{(0, \rho)} \rightarrow [\lambda, \tau]\rangle$, where $\alpha, \Upsilon, \rho, \tau \in
\mathbb{R}^+$ with $\rho\geq \alpha$, $\lambda \in \mathbb{R}^-$ and $m \in \mathbb{Z}^+$. Furthermore,
\begin{enumerate}
\item If ${\vec x_0}=\langle x_0^1, x_0^2, \ldots, x_0^k \rangle$ with
  $ x_0^1 < \alpha$, then $P(\chi; {\vec x_0})<\tau$, for all $\chi
  \in {\bar{\cal S}_{(0, \Upsilon)}}^m $ and for all $ {\vec x_0} \in
      {\bar{\cal S}_{(0, \rho)}}$.
\item $P(\chi; {\vec x_0}) \leq P(\chi; {\vec \phi})$, for all $\chi
  \in {\bar{\cal S}_{(0, \Upsilon)}}^m $ and for all $ {\vec x_0} \in
      {\bar{\cal S}_{(0, \rho)}}$.
\item $P(\langle{\vec \phi}, {\vec \phi}, \ldots, {\vec \phi}\rangle; {\vec \phi}) = 0$.
\end{enumerate}

\noindent
A neuron is said to {\em generate a
  spike} whenever $P(\cdot)=\tau$.

\end{definition}

\section{Feedforward Networks as Input-to-Output transformations}\label{sec:counterexample}
\begin{figure}
\begin{center}
  \includegraphics[width=8.6cm]{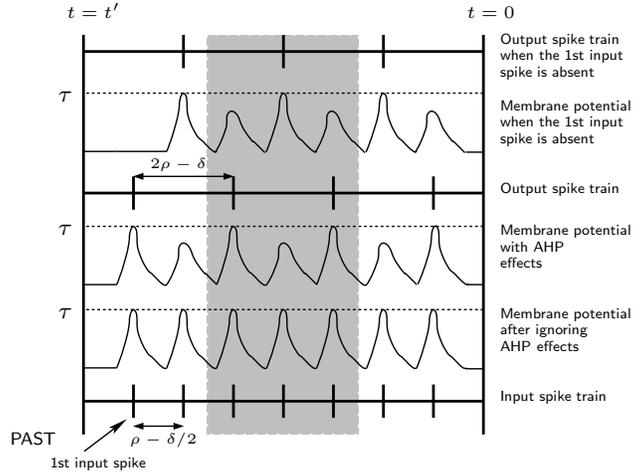}
\caption{This counterexample describes a single
neuron which has just one afferent synapse. Until time $t'$ in the
past, it received no input spikes. After this time, its input consisted of spikes that
arrived  every $\rho - \delta/2$ milliseconds, where $0<\delta\leq 2(\rho - \alpha)$. An input spike alone (if
there were no output spikes in the past $\rho$ milliseconds) causes this
neuron to produce an output spike. However, in addition, if there were an output
spike within the past $\rho$ milliseconds, the afterhyperpolarization (AHP) due to that spike is sufficient to bring
the potential below threshold, so that the neuron does not spike currently. We
therefore observe that if the first spike of the input spike-train is absent, then the output
spike-train changes drastically. Note that this change occurs no
matter how often the shaded segment in the middle is replicated, i.e. it
does not depend on how long ago the first spike occurred. Thus, the counterexample demonstrates that the membrane potential at any point in time may depend on the position of an input spike that occurred arbitrarily long time ago. Note that the input or the output pattern being
periodic and the two output patterns being phase-shifted is not a necessary ingredient of the counterexample; i.e. it is straightforward to construct a
(more complicated) counterexample that exhibits this same phenomenon
where neither the input spike-train nor the output spike-train are
periodic and where the two output spike patterns are not phase-shifted versions of each other.}
\label{fig:a}
\end{center}
\end{figure}

As discussed earlier, it is intuitively appealing to view feedforward
networks of neurons as transformations that map input
spike-trains to output spike-trains. In this section, we
seek to make this notion precise by clarifying in what sense, if at
all, these networks constitute the said transformations. It will turn
out that even single neurons cannot correctly be viewed as such transformations,
in general. In the next section, however, we show that under biologically-relevant spiking
regimes, we can salvage this view of feedforward
networks as spike-train to spike-train transformations.

Let us first consider the simplest type of feedforward network, namely
a single neuron.  Observe that our abstract neuron model does not
explicitly prescribe an output spike-train for a given input
spike-train ensemble. That is, recall from the previous section, that the
membrane potential of the neuron depends not only on the input spikes
received in the past $\Upsilon$ milliseconds, it also depends on the
output spikes produced by it in the past $\rho$
milliseconds. Therefore, knowledge of just input spike times in the
past $\Upsilon$ milliseconds does not uniquely determine the current
membrane potential (and therefore the output spike-train produced from
it). It might be tempting to then somehow use the fact that past
output spikes are themselves a function of input and output received
in the more distant past, and attempt to make the current membrane potential a
function of a bounded albeit larger ``window'' of past input spikes
alone. The simple counterexample described in Figure \ref{fig:a} shows
that this does not work. In particular, if we attempt to
characterize the current membrane potential of the neuron as a function of past
input spikes alone, the current membrane potential
 may depend on the position of an input spike that has
occurred arbitrarily long time ago in the past. To sum up, this counterexample proves that, without further restrictions, even a single neuron cannot be correctly viewed as a bounded-length spike-train to spike-train transformation.

This pessimistic prognosis notwithstanding, it may seem that if we
knew the infinite history of input spikes received by the neuron, we
should be able to uniquely determine its current membrane
potential. Unfortunately, the situation turns out to be even more dire
-- this turns out not to be the case. Before we demonstrate this, we must return to the issue of what it means for a neuron to {\em produce} an output spike-train when it receives a certain spike-train ensemble as input. That is, suppose the reader had an instantiation of our neuron model, which in this case would mean the values of $\Upsilon$, $\rho$ and $\tau$ and the membrane potential function $P(\cdot)$. Further, suppose the reader were given an input spike-train ensemble $\chi$ and told that the neuron ``produced'' the output spike-train ${\vec x_0}$ when driven by $\chi$. Then, all that the reader can do to verify this claim is to check if the given output spike-train is {\em consistent} with the input spike-train ensemble for the given neuron in the following sense. We would go to each point in time where the neuron spiked and plug into $P(\cdot)$  the input spikes in the past $\Upsilon$ milliseconds from $\chi$, and output spikes from the past $\rho$ milliseconds from ${\vec x_0}$ and check if the value of $P(\cdot)$ equals the threshold $\tau$. Likewise, for the time points where the output spike-train does not have a spike, we need to check that this value is less than the threshold. If the answers are in the affirmative for all time-points we can say that the given output spike-train is {\em consistent} with the given input spike-train ensemble with respect to the neuron in question. However, this still allows the possibility of more than one consistent output spike-train to exist for a given input spike-train ensemble, with respect to a given neuron. Indeed, we will demonstrate that this possibility can occur and therefore given the infinite history of input spikes received by the neuron, we cannot uniquely determine the output spike train produced. Before getting into the counterexample, for completeness, let us formally define this notion of {\em consistency}. Recall that $\langle t \rangle$ denotes a spike-train with a single spike at time instant $t$.

\begin{figure} 
\begin{center}
  \includegraphics[width=9.2cm]{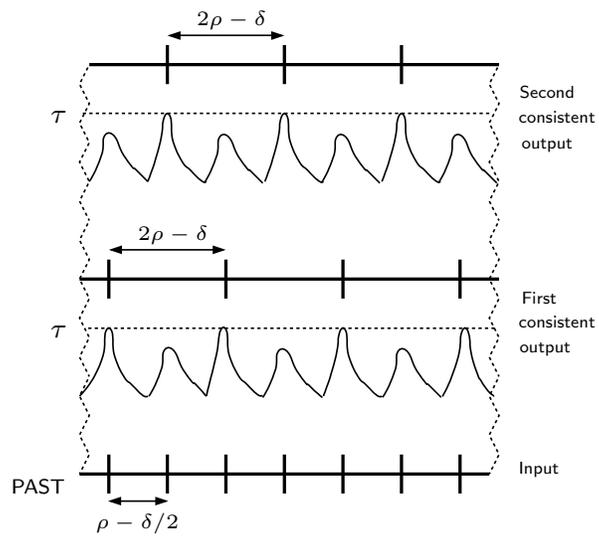}
\caption{The counterexample here is very similar to the one in Figure \ref{fig:a}, except
    that, instead of there being no input spikes before $t'$, we have
    an unbounded input spike-train ensemble, with the same periodic
    input spikes occurring since the infinite past. The neuron here has the exact same response properties as the one in Figure~\ref{fig:a}. Observe that both
    output spike-trains are consistent with this input, for each $t\in \mathbb{R}$. The corresponding membrane potential traces appear below each consistent output spike train.}
\label{fig:b}
\end{center}
\end{figure}

%\newpage

\begin{definition}
 An output spike-train ${\vec x_0}$ is said to be \emph{consistent}
 with an input spike-train ensemble $\chi$, with respect to a neuron ${\mathsf
   N}\langle\alpha, \Upsilon, \rho, \tau, \lambda, m, P:{\bar{\cal
     S}_{(0, \Upsilon)}}^m \times \bar{\cal S}_{(0, \rho)} \rightarrow
 [\lambda, \tau]\rangle$, if $\chi \in {\cal S}^m$ and the following holds. For every $t \in
 \mathbb{R}$, $\Xi_t {\vec x_0} = \langle t \rangle$ if and only if\\
 $P(\Xi_{(0,\Upsilon)}(\sigma_t(\chi)), \Xi_{(0,\rho)}(\sigma_t({\vec
   x_0})) =\tau$.

\end{definition}

\noindent
The question, therefore, is the following. For every (unbounded) input spike-train ensemble $\chi$, does there exist exactly one (unbounded) output spike
train ${\vec x_0}$, so that ${\vec x_0}$ is consistent with $\chi$ with respect to a given neuron ${\mathsf N}$?
As alluded to, the answer turns out to be in the negative. The counterexample in Figure~\ref{fig:b}
describes a neuron and an infinitely\footnote{The interested reader is referred to Appendix B for a discussion on the issue of infinitely-long input spike-trains in this context.} long input spike-train, which has
two consistent output spike-trains. 

The underlying difficulty in defining even single
neurons as spike-train to spike-train transformations, with both
viewpoints discussed above, is persistent dependence, in general, of current membrane
potential on ``initial state''. The way to circumvent this difficulty would be to impose additional restrictions which render such counterexamples untenable. For example, there is
the possibility of considering just a subset of input/output spike-trains, which  have the property of the current membrane potential
being independent of the input spikes beyond a certain time in
the past. Such a subset would certainly exclude the examples discussed in this
section. This would correspond to restricting our theory to a certain kind of spiking regime.

In the next section, we come up with a condition that, in effect, restricts
spike-trains to biologically-relevant spiking regimes and prove that this implies
independence as alluded to above. Roughly speaking, the condition is
that if a neuron has had a recent gap in its output spike-train equal
to at least {\em twice} its relative refractory period, then its
current membrane potential is independent of the input beyond the
relatively recent past. We show that this leads to the notion of
feedforward networks as spike-train to spike-train transformations to be
well-defined.

\section{The Gap Lemma and Criteria}\label{sec:GapLemma}
In this section, we devise a biologically well-motivated condition
that guarantees independence of current membrane potential from input
spikes beyond the recent past. This condition is used in constructing
a criterion for single neurons which when satisfied, guarantees a unique
consistent output spike-train and leads to the view of a neuron as a transformation that maps bounded-length input spike-trains to bounded-length output spike-trains. After this, similar criteria are
defined for feedforward networks, in general.

For a neuron, the way input spikes that happened sufficiently earlier affect current
membrane potential is via a causal sequence of output spikes, causal
in the sense that each output spike in the sequence had an effect on the membrane potential while the subsequent one
in the sequence was being produced and the input spike in question had an effect on the membrane potential, when the oldest output spike in the same sequence was produced. As a result, when an input spike is moved, this effect could propagate across time and cause the output spike train to change drastically. The condition in the Gap Lemma, in effect, seeks to break the causality in this causal chain.

\begin{figure}
\begin{center}
\includegraphics[width=8.6cm]{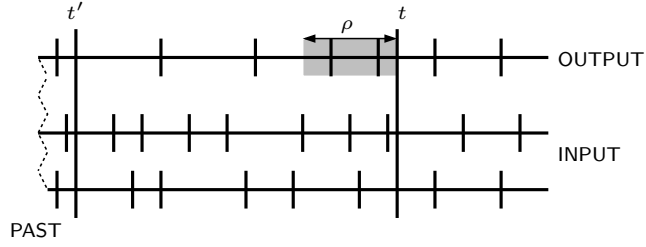}
  \caption{This figure illustrates the idea behind the Gap Lemma. Suppose there exists a neuron, with $\Upsilon$ and $\rho$ being the lengths of input and output windows respectively, that ``effects'' the transformation shown above. Let $(t'-t)\geq \Upsilon$. Suppose, the spikes in the shaded region, which is an
interval of length $\rho$ occurred at the exact same position, for all
input spike-train ensembles  that are identical in the range $[t, t']$, but have spikes occurring at arbitrary positions
older than time instant $t'$. Then, the membrane potential of that neuron at $t$ is identical in all those cases. This implies that the spikes in the shaded region are a function of exactly the input spikes in the interval $[t, t']$; in particular, they are independent of input spikes occurring before $t'$.}
\label{fig:c}
\end{center}
\end{figure}

Figure \ref{fig:c} elaborates the main idea behind the condition. Suppose there exists a neuron, with $\Upsilon$ and $\rho$ being the lengths of input and output windows respectively, that ``effects'' the transformation shown in Figure \ref{fig:c}. In a nutshell, if there was a guarantee that spike positions in an interval of length $\rho$ in the output spike train would remain invariant to changes in the past input spike-train ensemble, then this would break the aforementioned causal chain.

The question, of course, is what condition might guarantee such a situation. It turns out that a gap of length $2\rho$ in the output spike-train suffices, as the next lemma shows. That is, if the neuron effects a transformation with a $2\rho$ gap, say ending at $t$, present in the output, then for $t'$ being $\Upsilon+\rho$ milliseconds before $t$, such that no matter how input spikes older than $t'$ are changed, the latter half of the $2\rho$ gap is guaranteed to have no spikes in each case. Therefore, membrane potential starting at $t$, is the same in all such cases. $2\rho$ also turns out to be the smallest gap length for
which this works. Figure \ref{fig:d} offers some brief intuition on why a gap of length $2\rho$ suffices to guarantee independence. The technical details are in the following lemma. A formal proof is available in Appendix B.

\addtocounter{footnote}{-2}
\begin{figure}
\begin{center}
\includegraphics[width=8.6cm]{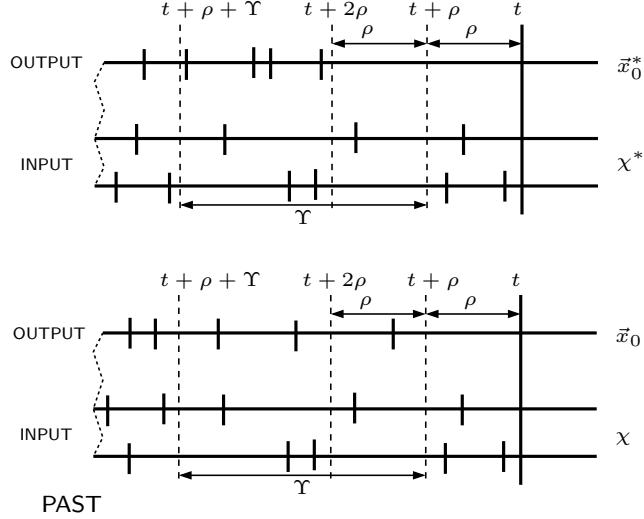}
\caption{This figure helps visualize the intuition behind why a gap of length $2\rho$ suffices to guarantee independence in the Gap Lemma. Suppose a neuron on receiving an input spike-train ensemble $\chi^*$ ``produces"\protect\footnotemark an output spike-train ${\vec x_0^*}$. Further, suppose, ${\vec x_0^*}$ has a gap of length $2 \rho$ ending at time instant $t$. Now let $\chi$ be some input spike-train ensemble, which is identical to $\chi^*$ in an interval of length $\Upsilon +\rho$ ending at $t$. Let ${\vec x_0}$ be the output spike-train "produced" by $\chi$. Then, the condition guarantees that ${\vec x_0}$ has a gap of length $\rho$ immediately preceding $t$. Here is why. When the neuron is being driven by $\chi^*$, clearly, the membrane potential is below threshold at each time instant $\rho$ milliseconds before $t$. At each such time instant, the neuron has no past output spikes $\rho$ milliseconds previously. Now, when the neuron is being driven by $\chi$ instead, there is no guarantee that the earlier half of the $2 \rho$ gap is preserved . Thus, at each time instant $\rho$ milliseconds before $t$, the neuron ``sees'' the same input spike-train ensemble $\Upsilon$ milliseconds previously as with $\chi^*$, but possibly some past output spikes $\rho$ milliseconds previously.  Therefore, it's membrane potential at each such time instant may be less than or equal to the corresponding value while the neuron was being driven by $\chi^*$, since, intuitively, the presence of recent efferent spikes could serve to afterhyperpolarize the membrane potential\protect\footnotemark. Thus, since the membrane potential was already below threshold in this time interval while the neuron was being driven by $\chi^*$, it is below the threshold, while the neuron is being driven by $\chi$ as well.}
\label{fig:d}
\end{center}
\end{figure}
\addtocounter{footnote}{-2}
\stepcounter{footnote}\footnotetext{For the sake of simplicity of exposition, assume there is exactly one consistent output spike-train. This is not a requirement as will become clear in the lemma.}
\stepcounter{footnote}\footnotetext{Formally, this follows from Axiom 2 in the definition of our abstract neuron.}

\begin{lemma}[\emph{Gap Lemma}]
Consider a neuron ${\mathsf N}\langle\alpha, \Upsilon, \rho, \tau,
\lambda, m, P:{\bar{\cal S}_{(0, \Upsilon)}}^m  \times \bar{\cal
  S}_{(0, \rho)} \rightarrow [\lambda, \tau]\rangle$, a spike-train ensemble $\chi^*$ of order $m$ and a spike-train ${\vec
  x_0}^*$ which has a gap in the interval $(t, t+2\rho)$, so that
${\vec x_0}^*$ is consistent with $\chi^*$, with respect to ${\mathsf
  N}$. Let $\chi$ be an arbitrary spike-train ensemble that is
identical to $\chi^*$ in the interval $(t, t+\Upsilon +\rho)$. 

Then, every output spike-train consistent with $\chi$, with respect to
${\mathsf N}$, has a gap in the interval $(t,t+\rho)$. Furthermore,
$2\rho$ is the smallest gap length in ${\vec x_0^*}$, for which this
is true.

\end{lemma}

The Gap Lemma has some ready implications as stated in the corollary below. A proof is available in Appendix B.

\begin{corollary}\label{corrgap}
Consider a neuron ${\mathsf N}\langle\alpha, \Upsilon, \rho, \tau,
\lambda, m, P:{\bar{\cal S}_{(0, \Upsilon)}}^m  \times \bar{\cal
  S}_{(0, \rho)} \rightarrow [\lambda, \tau]\rangle$, a spike-train
ensemble $\chi^*$ of order $m$ and a spike-train ${\vec x_0}^*$ which
has a gap in the interval $(t, t+2\rho)$ so that ${\vec x_0}^*$ is
consistent with $\chi^*$, with respect to ${\mathsf N}$. Then
\begin{enumerate}
\item Every ${\vec x_0}$ consistent with $\chi^*$, with respect to
  ${\mathsf N}$, has a gap in the interval $(t, t+\rho)$.\label{c11}
\item Every ${\vec x_0}$ consistent with $\chi^*$, with respect to
  ${\mathsf N}$, is identical to ${\vec x_0}^*$ in the interval $(-\infty, t+\rho)$, i.e. into the future after time instant $t+\rho$.\label{c12}
\item For every $t'$ more recent than $(t+\rho)$, the membrane potential at
  $t'$, is a function of spikes in $\Xi_{(t',t+\Upsilon
    +\rho)}(\chi^*)$.\label{c13}
\end{enumerate} 
\end{corollary}

The upshot of the Gap Lemma and its corollary is that whenever
a neuron goes through a period of time equal to twice its relative
refractory period where it has produced no output spikes it undergoes a ``reset'' in the sense that its membrane
potential from then on becomes independent of input spikes that
are older than $\Upsilon +\rho$ milliseconds before the end of the gap. 

Large gaps in the output spike-trains of neurons seem to be
extensively prevalent in the human brain. In parts of the brain where
the neurons spike persistently, such as in the frontal cortex, the
spike rate is very low (0.1Hz-10Hz) \citep{shep}. In contrast, the
typical spike rate of retinal ganglion cells can be very high but the
activity is generally interspersed with large gaps during which no
spikes are emitted \citep{lath}. %{\bf [SPRUCE UP BIO MOTIVATION SOME    MORE]}.

These observations motivate our definition of a criterion for input
spike-train ensembles afferent on single neurons. The criterion
stipulates that there be intermittent gaps of length at least twice the
relative refractory period in an output spike-train consistent with the input spike-train ensemble, with respect to the neuron in question. As we elaborate in a moment, the definition is set up so that for an input spike-train ensemble $\chi$ that satisfies a $T$-Gap criterion for a neuron, the membrane potential at any point in time is dependent on at most $T$
milliseconds of input spikes in $\chi$ before it.

\begin{definition}[Gap Criterion for a single neuron]
For ~$T \in \mathbb{R}^+$, a spike-train ensemble $\chi$ is said to satisfy a
$T${\emph -Gap Criterion}\footnote{Note that for sufficiently small values of $T$ (in relation to $\Upsilon$ and $\rho$), no $\chi$ may satisfy a $T$-Gap Criterion. This is deliberate formulation that will minimize notational clutter in forthcoming definitions.} for a neuron ${\mathsf N}\langle\alpha,
\Upsilon, \rho, \tau, \lambda, m, P:{\bar{\cal S}_{(0, \Upsilon)}}^m 
\times \bar{\cal S}_{(0, \rho)} \rightarrow [\lambda, \tau]\rangle$ if the following is true:
There exists a spike-train ${\vec x_0}$ with at least one gap of
length $2\rho$ in every interval of time of length $T-\Upsilon
+2\rho$, so that ${\vec x_0}$ is consistent with $\chi$ with respect to ${\mathsf N}$.
\end{definition}

Such input spike-train ensembles also have exactly one consistent output spike-train. The interested reader is directed to 
Proposition~1 in Appendix B for a formal statement and proof of this fact.
%Proposition~\ref{prop1} in Appendix B for a formal statement and proof of this fact.

\addtocounter{proposition}{1}

% The foll para can probably be made into a proposition with the result following from Corr 1(3).

\begin{figure}
\begin{center}
  \includegraphics[width=8.6cm]{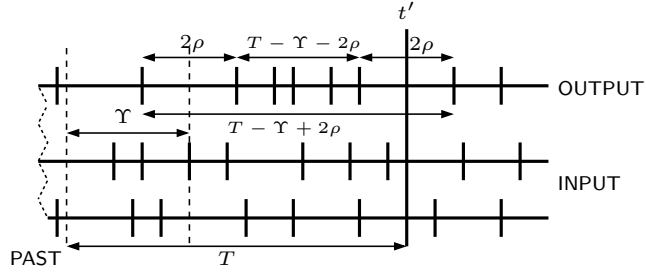}
\caption{Illustration demonstrating that for an input spike-train ensemble $\chi$ that satisfies a $T$-Gap criterion,
the membrane potential at any point in time is dependent on at most $T$
milliseconds of input spikes in $\chi$ before it. Owing to the $T$-Gap criterion the distance between the
end and start of any two consecutive gaps of length $2 \rho$ on the
output spike-train is at most $T-\Upsilon -2\rho$. Upto the earlier
half of a $2 \rho$ gap (whose latest point is denoted by $t'$) is
dependent on input corresponding to the previous $2\rho$ gap. It follows that the membrane potential at $t'$ depends only on
input spikes in the interval of length $T$ before it, as depicted,
owing to the Gap Lemma.}
\label{fig:e}
\end{center}
\end{figure}

For an input spike-train ensemble $\chi$ that satisfies a $T$-Gap
criterion for a neuron, the membrane potential at any point in time is
dependent on at most $T$ milliseconds of input spikes in $\chi$ before
it, as discussed in Figure \ref{fig:e}.

 With inputs that satisfy the $T$-Gap Criterion, here is what we need to do to
physically determine the current membrane potential, even if the
neuron has been receiving input since the infinite past: Start off
the neuron from an arbitrary state, and drive it with input that the
neuron received in the past $T$ milliseconds. The Gap Lemma guarantees that
the membrane potential we see now will be identical to the actual membrane
potential, since the membrane potential is guaranteed to have undergone a ``reset'' in the ensuing time.

The Gap Criterion we have defined for single neurons can be naturally
extended to the case of feedforward networks. The criterion is simply that the input
spike-train ensemble to the network is such that every neuron's input obeys a
 scaled Gap criterion for single neurons. Figure~\ref{fig:g} explains the idea. Formally, the definition proceeds inductively, starting with neurons of depth 1. 

\begin{definition}[Gap Criterion for a feedforward network]
An input spike-train ensemble $\chi$ is said to satisfy a $T${\emph
  -Gap Criterion} for a feedforward network if each neuron in the network
satisfies a $(\frac{T}{d})$-Gap Criterion, when the network is driven
by $\chi$, where $d$ is the depth of the acyclic network.
\end{definition}

\begin{figure}
\begin{center}
  \includegraphics[width=8.5cm,height=3.2cm]{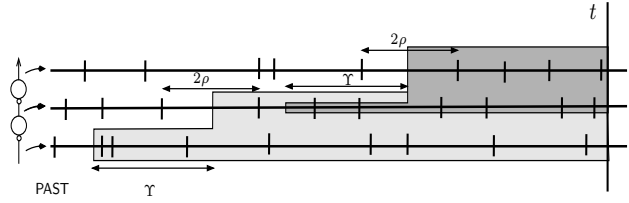}
\caption{Schematic diagram illustrating how the Gap criterion works for the simple two-neuron network on the left. The membrane potential of the output neuron at $t$ depends on input received from the ``intermediate'' neuron, as depicted in the darkly-shaded region, owing to the Gap Lemma. The output of the intermediate neuron in the darkly-shaded region, in turn, depends on input it received in the lightly-shaded region. Thus, transitively, membrane potential of the output neuron  at $t$ is dependent at most on input received by the network in the lightly-shaded region.}
\label{fig:g}
\end{center}
\end{figure}

As with the criterion for the single neuron, the membrane potential of
the output neuron at any point is dependent on at most $T$ milliseconds of
past input, if the input spike-train ensemble to the feedforward network
satisfies a $T$-Gap criterion. Additionally, the output spike-train is unique. 
Lemma~2  and its proof in Appendix B make precise these facts. 
%Lemma~\ref{lemma2:gap}  and its proof in Appendix B make precise these facts. 

\addtocounter{lemma}{1}

We thus find ourselves at a juncture where questions we initially sought to ask can
be posed in a self-consistent manner. So, looking back at the big picture, we had initially wished to view feedforward networks as transformations that mapped bounded-length input spike-trains to bounded-length output spike trains. However, we found that this notion was not always well-defined. We then showed that if we restrict the set of input spike-trains so they satisfied certain criteria,  one can correctly speak of output spike-trains that such inputs are mapped to, by the feedforward network in question. We also argued that this restricted set of spike-trains encompasses biologically-relevant spiking regimes. Thus, feedforward networks can be seen as transformations that map this restricted set of input spike-trains to output spike-trains. Indeed, this will be the sense in which feedforward networks are treated as transformations. Next, we formalize these observations and define some notation.

{\bf Notation.} Given a feedforward network ${\cal N}$, let $ {\cal G}^T_{\cal N}$ be
the set of all input spike-train ensembles that satisfy a $T$-Gap
Criterion for ${\cal N}$. Let $ {\cal G}_{\cal N} = \bigcup_{T \in
  \mathbb{R}^+}  {\cal G}^T_{\cal N} $. Therefore, every feedforward
network ~${\cal N}$ induces a transformation ${\cal T}_{\cal N}: {\cal
  G}_{\cal N} \rightarrow {\cal S}$ that maps each spike-train ensemble in $ {\cal G}_{\cal N}$ to a unique output spike
train in the set of spike-trains ${\cal S}$. Suppose ${\cal G'} \subseteq {\cal G}_{\cal
  N}$. Then, let  ${\cal T}_{\cal N}|_{{\scriptscriptstyle{\cal G'}}}: {\cal
  G'} \rightarrow {\cal S}$ be the map defined as  ${\cal T}_{\cal N}|_{{\scriptscriptstyle{\cal G'}}}(\chi)={\cal T}_{\cal N}(\chi)$, for all $\chi \in {\cal G'}$.

The Gap Criteria are very general and biologically well-motivated. However,
given a neuron or a feedforward network, there does not appear to be an
easy way to characterize all the input spike-train ensembles that satisfy a
certain Gap Criterion for it. That is, for a given neuron, whether an input spike-train ensemble satisfies a Gap Criterion for it seems to depend intimately on the exact form of its membrane potential function. As a result, a spike-train ensemble that satisfies a Gap criterion for one neuron may not satisfy any Gap Criterion for another neuron. For a feedforward network, the problem becomes even more difficult, since intermediate
neurons must satisfy Gap Criteria, and also produce output spike-trains that satisfy Gap Criteria for neurons further downstream. Furthermore, in order to compare transformations effected by two different networks, we need to study inputs that satisfy some Gap criterion for both of them, for otherwise, the notion of a transformation may no longer hold. Now, we sought to ask what transformations {\em all} feedforward networks with a certain architecture could not do. For this, we need to characterize inputs that satisfy a Gap Criterion for all the networks involved, which seems to be an even more intractable problem.

This brings up the question of the existence of another criterion
according to which the set of spike-train ensembles is easier to
characterize and is {\em common} across different networks. Next, we
propose one such criterion and show that it consists of spike-train
ensembles which are a subset of those induced by the Gap criteria for
all feedforward networks. Loosely speaking, these are input spike-train
ensembles which, before a certain time instant in the past, have had
no spikes. The spike-train ensembles satisfying the said criterion,
which we call the Flush criterion, allow us to sidestep the difficult
issues just discussed. While this is a purely
theoretical construct with no claim of biological relevance, in
Section~\ref{sec:compl}, we prove that there is no loss by restricting
ourselves to the Flush criterion. That is, not only is a 
result proved using the Flush criterion applicable with the Gap
criterion, {\em every} result true with the Gap criterion can be
proved by using the Flush criterion exclusively.

\section{Flush Criterion}\label{sec:Flush}
The idea of the Flush Criterion is to force the neuron to produce no
output spikes for sufficiently long so as to guarantee that a Gap
criterion is being satisfied. This is done by having a
semi-infinitely long interval with no input spikes. This ``flushes'' the
neuron by bringing it to the resting potential and keeps it there for
a sufficiently long time, during which it produces no output spikes. In a feedforward network, the flush is
propagated so that all neurons have had a sufficiently long gap in
their output spike-trains. Observe that the Flush Criterion is not
defined with reference to any feedforward network and is just a property of the spike-train ensemble. We make this notion precise below.

\begin{definition}[Flush Criterion]
A spike-train ensemble $\chi$ is said to satisfy a {\em $T$-Flush Criterion},
if all its spikes lie in the interval $(0,T)$, i.e. it has no spikes
upto time instant $T$ and since time instant 0.
\end{definition}

It turns out that an input spike-train ensemble to a neuron that satisfies a
Flush criterion also satisfies a Gap criterion.  The technical details along with a proof are in 
Lemma~3 in Appendix B.
%Lemma~\ref{lemma3} in Appendix B.

\addtocounter{lemma}{1}

Likewise, an input spike-train ensemble to a feedforward network satisfying a Flush
criterion also satisfies a Gap criterion for that network, as elaborated in 
Lemma~4 which is available in Appendix B with a proof.
%Lemma~\ref{lemma4} which is available in Appendix B with a proof.

\addtocounter{lemma}{1}

The Flush criterion is a construct made for mathematical expedience and prima facie does not have any biological relevance. It is a network-independent criterion which enables us to circumvent difficulties that working with the Gap criterion entailed. It will soon become clear why it is a useful construction, when we show that it is equivalent to the Gap criterion insofar as the questions we seek to ask are concerned.

\section{Transformational Complexity}\label{sec:compl}
Having laid the groundwork, in this section, we set up a definition that will allow us to ask if there exists a transformation that no network of a certain architecture could effect that a network of a different architecture could. It is convenient to formulate the definition in the following terms. Given two classes\footnote{The
classes of networks could correspond to ones that contain all networks with specific network architectures, although for the
purpose of the definition, there is no reason to require this to be
the case.} of networks with
the second class encompassing the first, we ask if there is a
network in the second class whose transformation cannot be performed by
any network in the first class. That is, does the second class possess a larger repertoire of transformations than the first, giving it {\em more complex} computational capabilities?  

\begin{definition}[Transformational Complexity]\label{defn:compl}
Let $\Sigma_1$ and $\Sigma_2$ be two sets of feedforward networks, each
network being of order $m$, with $\Sigma_1 \subseteq \Sigma_2$. Define
${\cal G}_{12} = \bigcap_{{\cal N} \in \Sigma_2} {\cal
  G}_{\cal N}$. The set $\Sigma_2$ is said to be {\em more complex than}
$\Sigma_1$, if there exists an ${\cal N}' \in \Sigma_2$ such that for all
${\cal N} \in \Sigma_1, {\cal T}_{{\cal N}'}|_{{\scriptscriptstyle{\cal G}_{12}}} \neq {\cal
  T}_{{\cal N}}|_{{\scriptscriptstyle{\cal G}_{12}}}$.
 
\end{definition}

%{\bf[ MOTIVATe ${\cal G}_{12}$]}

\noindent
A couple of remarks about the definition above are in order. Firstly, $\Sigma_1$ being a proper subset of $\Sigma_2$, does not necessarily imply that the that the set of transformations effected by  networks in $\Sigma_1$ is also a proper subset of those effected by $\Sigma_2$. In particular, it could be the case that the set of transformations effected by $\Sigma_1$ is exactly the same as that effected by $\Sigma_2$, even though $\Sigma_1$ is a proper subset of $\Sigma_2$. Indeed, this is what is demonstrated by the result of Section 9, which shows in the context of the present neuron model that even though the set of depth-two feedforward networks is a strict subset of the set of all feedforward networks, both these sets effect the same class of transformations, namely those that are causal, time-invariant and resettable. Secondly, observe that while comparing a set of networks, we restrict ourselves to
inputs for which all the networks satisfy a certain Gap Criterion
(though, not necessarily for the same $T$), so that the notion of a
transformation is well-defined on the input set, for all networks
under consideration. Note also that ${\cal G}_{12}$ is always a nonempty set, because ${\cal G}_{12}$ contains within it all inputs satisfying the Flush criterion. Henceforth, for brevity, any result that establishes a
relationship of the form defined above is called a {\em complexity
result.} Before we proceed, we introduce some useful notation. 

{\bf Notation.} Let the set of spike-train ensembles
of order $m$ that satisfy the T-Flush criterion be
${\cal F}_m^T$.  Let ${\cal F}_m=\bigcup_{T \in \mathbb{R}^+} {\cal
  F}_m^T $. What we have established in the previous section is that ${\cal
  F}_m \subseteq {\cal G}_{\cal N}$, for every feedforward network ${\cal N}$ of order $m$.

Next, we show that if one class of networks is more complex than
another, then inputs that satisfy the Flush Criterion are both necessary and sufficient
to prove this. That is, to prove this type of complexity result, one
can work exclusively with Flush inputs without losing any
generality. This is not obvious because Flush inputs form a subset of
the more biologically well-motivated Gap inputs. The next lemma formalizes this equivalence. Note that the statement of the lemma is substantially identical to that of Definition~\ref{defn:compl}, except that the input spike-train ensembles in the lemma below satisfy the Flush criterion, as opposed to the ones in  Definition~\ref{defn:compl} which satisfy ${\cal G}_{12}$, the set of input spike-train ensembles that satisfy a Gap Criterion for all the networks under consideration.

\begin{lemma}[Equivalence of Flush and Gap Criteria with respect to Transformational Complexity]\label{thm:GFeq}
Let $\Sigma_1$ and $\Sigma_2$ be two sets of feedforward networks, each
network being of order $m$, with $\Sigma_1 \subseteq \Sigma_2$. Then,
$\Sigma_2$ is more complex than $\Sigma_1$ if and only if $\exists
{\cal N}' \in \Sigma_2$ such that $\forall {\cal N} \in \Sigma_1,
{\cal T}_{{\cal N}'}|_{{\scriptscriptstyle{\cal F}_{m}}} \neq {\cal T}_{{\cal N}}|_{{\scriptscriptstyle{\cal F}_{m}}}$.
\end{lemma}
\begin{proof}[Proof sketch]
A full proof is available in Appendix B; here we sketch the intuition behind the proof. 

Showing that Flush inputs are sufficient is the easier half of the proof. If a complexity result can be shown using Flush inputs, it follows that it holds for Gap inputs as well, since ${\cal F}_m \subseteq {\cal G}_{12}$. To show that the existence of Flush inputs is necessary, we assume a
complexity result proved using Gap inputs and construct Flush inputs
such that the result can be shown using those Flush inputs alone. Now
suppose ${\cal N}' \in \Sigma_2$ be the network such that no network in $\Sigma_1$ effects the same transformation as ${\cal N}'$, when the domain is restricted to the set ${\cal G}_{12}$. Now, consider arbitrary ${\cal N} \in \Sigma_1$. There must exist a $\chi \in {\cal G}_{12}$ such that ${\cal T}_{{\cal N}'}|_{{\scriptscriptstyle{\cal F}_{m}}}(\chi) \neq {\cal T}_{{\cal N}}|_{{\scriptscriptstyle{\cal F}_{m}}}(\chi)$. By definition, this $\chi$ satisfies a $T_1$-Gap Criterion for
${\cal N}$ and a $T_2$-Gap Criterion for ${\cal N}'$. Let $T=\max(T_1, T_2)$. The claim is that if $\chi$ is cut up into ``chunks'' of length $2T$, where each ``chunk'' satisfies a 2T-Flush criterion, then ${\cal N}$ and ${\cal N}'$ will map at least one of those chunks to different output spike trains, since the output in the latter half of the chunk is identical to that produced by the corresponding segment of $\chi$. This process of ``cutting up'', when ``completed'' for each ${\cal N} \in \Sigma_1$ yields a subset of Flush inputs, using which the complexity result can be established.
\end{proof}

Assured by this theoretical guarantee that there is no loss of
generality by doing so, we will henceforth only work with inputs
satisfying the Flush Criterion, while faced with the task of proving
complexity results. This buys us a great deal of mathematical expedience at no cost. From now on, unless qualified otherwise, when we speak of a {\em transformation}, we mean a map of the form ${\cal T}: {\cal F}_{m} \rightarrow {\cal S}$ that maps the set of Flush input spike-train ensembles to the set of output spike-trains.

\section{Complexity results}\label{sec:compl_results}

In this section, we establish some complexity results. First, we show
that there exist spike-train to spike-train transformations that no
feedforward network can effect. Next, we show a transformation that no
single neuron can effect but a network consisting of two neurons
can. After this, we prove a result which shows that a class of
architectures that share a certain structural property also share in
their inability in effecting a particular class of
transformations. Notably, while this class of architectures has
networks with arbitrarily many neurons, we show a class of networks
with just two neurons which can effect this class of
transformations. The interested reader is directed to Appendix
B for some technical remarks concerning the mechanics of proving complexity results that are not central to the exposition
here.

\addtocounter{footnote}{-1}
\begin{figure}
\begin{center}
\subfigure[Example of a transformation that no feedforward network\protect\footnotemark can effect. The shaded region is replicated over, to obtain mappings for larger and larger values of $T$.]{\label{fig:k}\includegraphics[width=8.2cm]{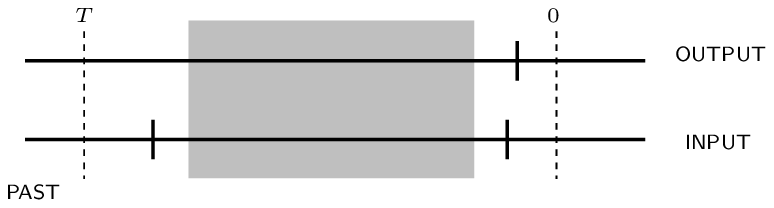}}
\subfigure[A transformation that no single neuron can effect, that a network with two neurons can.]{\label{fig:i}\includegraphics[width=8.2cm]{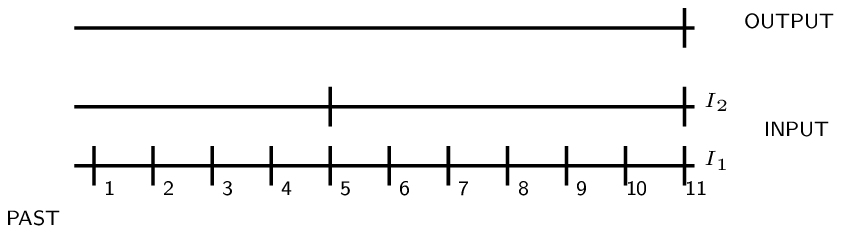}}
\caption{}
\end{center}
\end{figure}
\addtocounter{footnote}{-1}
\stepcounter{footnote}\footnotetext{Recall that the neurons considered in this work are deterministic.}

Before establishing complexity results, we point out that it is
straightforward to construct a transformation that cannot be effected by any
feedforward network. One of its input spike-train ensembles with the
prescribed output is shown in Figure~\ref{fig:k}. For larger ~$T$, the shaded region is
simply replicated over and over again. Informally, the reason this
transformation cannot be effected by any network is that, for any
network, beyond a certain value of ~$T$, the shaded region tends to act
as a ``flush'', erasing ``memory'' of the first input spike. When the network receives another input spike, it is in the exact same ``state'' it was
when it received the first input spike, and therefore cannot produce an
output spike after the second input spike. 

Next, we prove that the set of feedforward 
networks with at most two neurons is more complex than the set of
single neurons. The proof is by prescribing a transformation which
cannot be done by any single neuron. We then construct a network with
two neurons that can indeed effect this transformation. Note that in the statement of the theorem below, $m$ stands for the number of input spike trains.

\begin{theorem} \label{thm:1vs2}
Suppose $m\geq 2$. Let $\Sigma$ be the set of feedforward networks with at most two neurons that each receive an input spike-train ensemble of order $m$. Then, $\Sigma$ is
more complex than the set of single neurons of order $m$.
\end{theorem}

\begin{proof}
We first prescribe a transformation, prove that it cannot be effected by a single neuron and then construct a two-neuron network and show that it can indeed effect the same transformation.

We first prove the result for $m=2$ and later indicate how it can be
 extended for larger values of $m$. Let the two input
spike-trains in each input spike-train ensemble, which satisfies a
Flush Criterion be $I_1$ and $I_2$.  Figure \ref{fig:i} illustrates
the transformation. Informally, $I_1$ has regularly-spaced spikes
starting after time instant $T$ until $0$. $I_2$ has two spikes,
with the first one, loosely speaking, in the ``middle'' of $(0, T)$ and
the second one at the end, i.e. right before time instant $0$. An output spike
is always prescribed after the second spike in $I_2$ occurs, and not
elsewhere. For larger $T$, the number of spikes on $I_1$ increases so
as to maintain the same regular spacing; $I_2$, in contrast, still has just
two spikes, the first one roughly in the middle and the second in the
end. For the sake of exposition, we call the distance between
consecutive spikes on $I_1$, one time unit and we number the spikes of $I_1$
with the first spike being the oldest one.

More precisely, the transformation is prescribed for a subset of ${\cal
  F}_{m}$, whose elements are indexed by $i=1, 2, \cdots$.  Figure \ref{fig:i} illustrates the transformation,
for $i=2$. The $i$th input spike-train
ensemble in this subset satisfies a $T$-Flush criterion, where
$T=4i+3$ time units. In the $i$th spike-train ensemble, $I_2$ has spikes at time
instants at which spike numbers $2i+1$ and $4i+3$ occur in
$I_1$. Finally, the output spike-train corresponding to the $i$th
input spike-train ensemble has exactly one spike after\footnote{Strictly speaking, the output spike happens at $4i+3+\epsilon$, where $\epsilon>0$ is a small real number. Henceforth whenever we say an output spike is {\em after} a certain time instant, we mean it in this sense.} the time instant
at which $I_1$ has spike number $4i+3$.

Next, we prove that the transformation prescribed above cannot be
effected by any single neuron. For the sake of contradiction, suppose
it can, by a neuron with associated $\Upsilon$ and $\rho$. Let
$\max(\Upsilon, \rho)$ be bounded from above by $k$ time units. We
show that for $i\geq\lceil\frac{k}{2}\rceil$, the $i$th input spike-train ensemble
cannot be mapped by this neuron to the prescribed output spike
train. For $i=\lceil\frac{k}{2}\rceil$, consider the membrane potential of the neuron after the time instants
corresponding to the $(k+1)$th spike number and $(2k+3)$rd spike
number of $I_1$. At each of these corresponding time instants, the input received in
the past $k$ time units and the output produced by the neuron in the
past $k$ time units are the same. Therefore, the neuron's membrane
potential must be identical as well. However, the transformation prescribes no
spike in one of the first time instants and a spike in the second, which is a
contradiction. It follows that no single neuron can effect the
prescribed transformation.

\begin{figure}
\begin{center}
  \includegraphics[width=8.4cm]{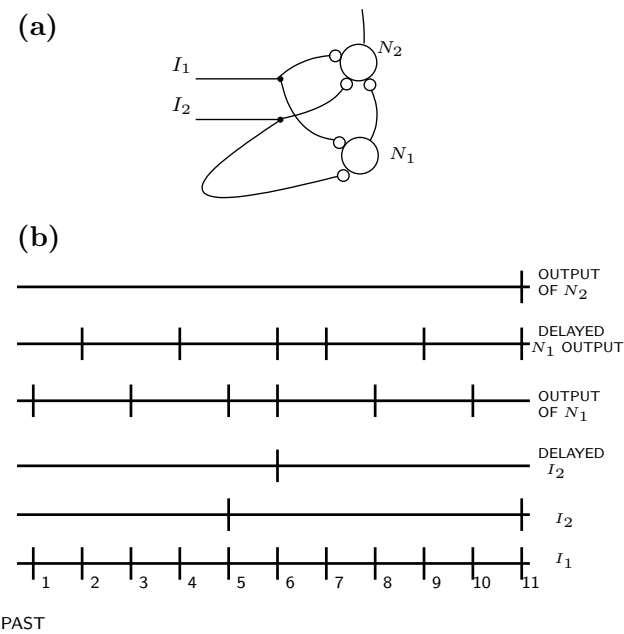}
\caption{(a) The network that can effect the transformation described in Figure \ref{fig:i}. (b)~~Figure describing the operation of this network.}
\label{fig:hj}
\end{center}
\end{figure}

We now construct a two-neuron network which can carry out the
prescribed transformation. The network is shown in Figure \ref{fig:hj}(a). $I_1$ and $I_2$ arrive instantaneously at $N_2$. $I_1$
arrives instantaneously at $N_1$ but $I_2$ arrives at $N_1$ after a
delay of $1$ time unit. Spikes output by $N_1$ take one time unit to
arrive at $N_2$, which is the output neuron of the network. The
functioning of this network for $i=2$ is described in Figure \ref{fig:hj}(b). The generalization for larger $i$ is
straightforward. All inputs are excitatory. $N_1$ is akin to the
neuron described in Figure \ref{fig:a}, in that
while the depolarization due to a spike in $I_1$ causes potential to
cross threshold, if, additionally, the previous output spike happened one time unit ago, the associated
hyperpolarization is sufficient to keep the membrane potential below
threshold now. However, if there is a spike from $I_2$ also at the
same time as from $I_1$, the depolarization is sufficient to cause an output spike,
irrespective of if there was an output spike one time unit ago. The
$\Upsilon$ corresponding to $N_2$ is shorter than $1$ time
unit. Further, $N_2$ produces a spike if and only if all three of its
afferent synapses receive spikes at the same time. In the figure,
$N_1$ spikes after times $1, 3, 5$. It spikes after $6$ because it received
spikes both from $I_1$ and $I_2$ at that time instant. Subsequently,
it spikes after $8$ and $10$. The only time wherein $N_2$ received spikes
at all three synapses at the same time is at $11$, after which is the
prescribed time for the output spike. The generalization for larger $i$ is
straightforward.

For larger $m$, to construct a transformation that cannot be done by a single neuron but can be, by a two-neuron network, one can just have the same input as $I_1$ or $I_2$ on the extra input spike
trains and the same proof generalizes easily. 
\end{proof}

The previous result might seem to suggest that the more the number of neurons (and connections between them)
the larger the variety of transformations possible. The next
complexity result demonstrates, on the contrary,  that the structure of the network
architecture is crucial. That is, we can construct network
architectures with arbitrarily large number of neurons which cannot
perform transformations that a two-neuron network with simple neurons can. 

First, we define the structural property that characterizes this class of architectures.

\begin{definition}[Path-plural Network]
A feedforward network of order $m$ is called {\em path-plural} if 
for every set of $m$
 paths, where the $i$th path starts at $i$th input vertex and ends at
 the output vertex, the intersection of the $m$ paths is exactly the
 output vertex.
\end{definition}

Every feedforward network in which all the inputs aren't afferent on every neuron, must have embedded within it a path-plural network. For this reason, path-plural networks are an important and ubiquitous class of feedforward networks. How large such networks are in the brain remains to be seen, and this will become clearer as we get more and more data from the connectomics efforts. But, it is conceivable that such networks exist in feedforward pathways that that converge onto networks that, for example, integrate information from multiple sensory modalities.

\noindent
We now state and prove the complexity result.

\begin{theorem}
 For $m\geq 3$, let $\Sigma_1$ be the set of all path-plural feedforward networks
 of order $m$. Let $\Sigma_2$ be the union of
 $\Sigma_1$ with the set of all two-neuron feedforward networks of order $m$. Then, $\Sigma_2$ is more complex than
 $\Sigma_1$.
\end{theorem}

\begin{proof}
We first prescribe a transformation, prove that it cannot be effected by any network in $\Sigma_1$ and then construct a two-neuron network and show that it can indeed effect the same transformation.

We prove the theorem for $m=3$; the generalization for larger $m$ is
straightforward.  The following transformation is prescribed for
$m=3$. Let the three input spike-trains in each input spike
train ensemble, which satisfies a Flush Criterion be $I_1$, $I_2$ and
$I_3$. As before, we will use regularly spaced spikes; we call the
distance between two such consecutive spikes one time unit and number
these spike time instants with the oldest being numbered 1; we call
this numbering the spike index. Again, the transformation is prescribed for a subset of ${\cal
  F}_{m}$, whose elements are indexed by $i=1, 2, \cdots$.  Figure \ref{fig:transf}
illustrates the transformation for $i=2$. The $i$th input spike-train ensemble in
the subset satisfies a $T$-Flush Criterion for $T
= 4im$ time units. The first $2i$ time units have spikes on $I_2$
spaced one time unit apart, the next $2i$ on $I_3$ and so forth. In
addition, at spike index $2im$, $I_m$ has a single spike. The input spike pattern
from the beginning is repeated once again for the latter $2im$ time units. The prescribed output spike-train has
exactly one spike after spike index $4im$.

\begin{figure}
\begin{center}
  \includegraphics[width=8.7cm]{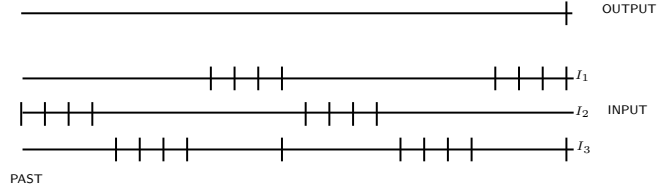}
\caption{A transformation that no feedforward network of order $3$ with a path-plural architecture can effect.}
\label{fig:transf}
\end{center}
\end{figure}

Next we prove that the transformation prescribed above cannot be
effected by any network in $\Sigma_1$. For the sake of contradiction,
assume that there exists a network ${\cal N} \in \Sigma_1$ that can effect the transformation. Let
$\Upsilon$ and $\rho$ be upper bounds on the same parameters over all
of the neurons in ${\cal N}$ and let $d$ be the depth of ${\cal
  N}$. By construction of $\Sigma_1$, every neuron in ${\cal N}$ that
is afferent on the output neuron receives input from at most $m-1$ of
the input spike-trains; for, otherwise there would exist a set of $m$
paths, one from each input vertex to the output neuron, whose
intersection would contain the neuron in question. The claim, now, is
that for $i> \frac{\Upsilon d}{2} + \rho$, the output neuron of ${\cal
  N}$ has the same membrane potential at spike index $2im$ and
$4im$, and therefore either has to spike at both those instants or
not. Intuitively, this is so because each neuron afferent on the
output neuron receives a ``flush'' at some point after $2im$, so that
the output produced by it $\Upsilon$ milliseconds before time index $2im$
and $\Upsilon$ milliseconds before time index $4im$ are the same. This is straightforward to verify.

\begin{figure}
\begin{center}
  \includegraphics[width=8.7cm]{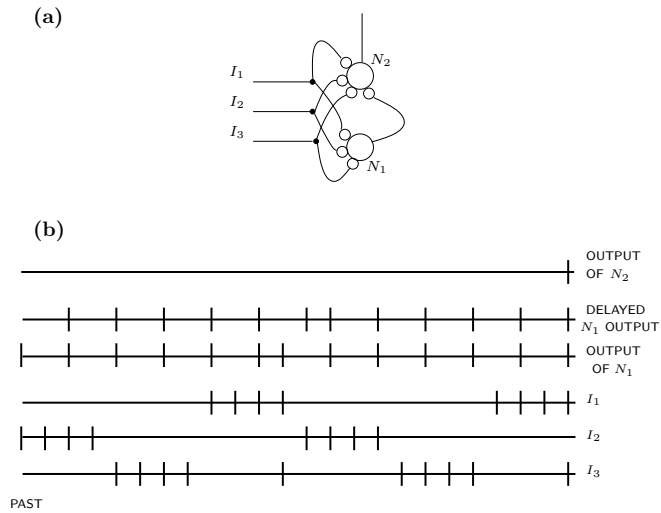}
\caption{(a) Network that can effect the transformation described in Figure \ref{fig:transf}. (b)~~Figure describing the operation of this network.}
\label{fig:lm}
\end{center}
\end{figure}

We now construct a two-neuron network that can effect this
transformation. The construction is similar to the one used in Theorem
\ref{thm:1vs2}. For $m=3$, the network is shown in Figure \ref{fig:lm}.  $I_1$, $I_2$ and $I_3$ arrive instantaneously at $N_1$ and
$N_2$. Spikes output by $N_1$ take two time units to arrive at $N_2$,
which is the output neuron of the network. The functioning of this
network for $i=2$ is described in Figure \ref{fig:lm}(b). The
generalization for larger $i$ is straightforward. All inputs are
excitatory. $N_1$ is akin to the the neuron $N_1$ used in the network in
Theorem~\ref{thm:1vs2} except that that periodic input may arrive from
any one of $I_1$, $I_2$ or $I_3$. As before, if two input spikes
arrive at the same time, as in spike index $2im$, the depolarization is
sufficient to cause an output spike in $N_1$, irrespective of if there was an
output spike one time unit ago. Again, the $\Upsilon$ corresponding to
$N_2$ is shorter than $1$ time unit and $N_2$ produces a spike if and
only if three of its afferent synapses receive spikes at the same
time instant. As before, the idea is that at time $2im$, $N_2$, receives two
spikes, but not a spike from $N_1$, since it is ``out of
sync''. However, at time $4im$, additionally, there is a spike from
$N_1$ arriving at $N_2$, which causes $N_2$ to spike.
\end{proof}

To conclude, what we have demonstrated in this section is that, for certain classes of networks, just by knowing the architecture of the network, we can rule out computations that the network could be doing. All we assumed was that the neurons in the network satisfy a small number of elementary properties; notably these results do not require knowledge of detailed physiological properties of the neurons in the network. This, in itself, is somewhat surprising due to the intuitively-appealing expectation that network structure may not impose as strong a constraint as neurophysiology insofar as the computational ability of a network is concerned. In the next section, however, we show that this intuition is sound in some cases by proving that there are limits to the constraints imposed by network structure in the presence of very limited information on the physiology.

\section{Limits to constraints imposed by network structure}\label{sec:depth2}

The main thrust of this work, thus far, has been in demonstrating that connectomic constraints do indeed restrict the computational ability of certain networks, even when we do not assume much about the physiological properties of their neurons. As one might expect, we should be able to get better mileage, so to speak, if we had more elaborate information on the response properties of the individual neurons. Conversely, it is logical to expect that there might be fundamental limits to what can be said about the computational properties of networks, given very limited knowledge of the neurophysiology of its neurons. In this section, we prove this to be the case. In particular, we show that the small set of assumptions made about our model neurons lead to the absence of connectomic constraints on computation for the class of feedforward networks of depth equal to two. More precisely, it turns out that there does not exist a transformation that cannot be performed by any network of depth two\footnote{equipped with instances of our model neurons} that in turn can be effected by another network (of a different architecture). What this result implies is that one {\em needs} to make further assumptions on the properties obeyed by the model neurons, before connectomic constraints on this class of networks appear.

So, how does one prove that there does not exist a transformation that cannot be performed by any network of depth two that in turn can be effected by another network? Equivalently, we need to prove that given an arbitrary feedforward network, there exists a feedforward network of depth two that effects {\em exactly} the same transformation.

The difficulty in proving that every feedforward network, having arbitrary
depth, has an equivalent network of depth two, appears to be in
devising a way of ``collapsing'' the depth of the former network,
while keeping the effected transformation the same. Our proof actually
does not demonstrate this head-on, but instead proves it to be the case
indirectly. The broad attack is the following: Consider the set of transformations spanned by the set of all feedforward networks. Recall that this is a proper subset of the set of all transformations, since we had shown a transformation that no feedforward network could effect. The idea is to start off with a certain ``nice'' subset of the set
of all transformations and show that every transformation effected by feedforward networks certainly lies within this subset. Thereafter, we prove, by providing
a construction, that every transformation in this ``nice''
subset can in fact be effected by a feedforward network of depth two\footnote{As a by-product, the proof also ends up providing a complete characterization of the set of transformations spanned by the set of all feedforward networks equipped with neurons of the present abstract model, which turns out to be exactly this ``nice'' set.}. Together, this implies that, for every transformation that can be effected by a feedforward network, there exists a feedforward network of depth two that can effect exactly that transformation. 

The interested reader is directed to Appendix C, which is a 24-minute video that provides an intuitive outline of the results in this section using animations.

\subsubsection*{Technical structure of the proof}
\noindent
The main theorem that we prove in this section is the following.\\

\begin{theorem}
If ~${\cal T}:{\cal F}_m\rightarrow{\cal S}$ can be effected by a feedforward network, then it can be effected by a feedforward network of depth two.
\end{theorem}
\addtocounter{theorem}{-1}

\noindent
This theorem follows from the following two lemmas which are proved in the two subsections that follow:

\begin{lemma}
If ~${\cal T}:{\cal F}_m\rightarrow{\cal S}$ can be effected by a feedforward network, then ${\cal T}(\cdot)$ is causal, time-invariant and resettable.
\end{lemma}

\begin{lemma}
If ~${\cal T}:{\cal F}_m\rightarrow{\cal S}$ is causal, time-invariant and resettable, then it can be effected by a feedforward network of depth two.
\end{lemma}

\addtocounter{lemma}{-2}

\subsection{Causal, Time-Invariant and Resettable Transformations}

In this section, we first define notions of causal, time-invariant and resettable transformations\footnote{Recall that when we say transformation, without further qualification, we mean one, of the form ${\cal T}:{\cal F}_m\rightarrow{\cal S}$.}. Transformations that are causal, time-invariant and resettable form a strict subset of the set of all transformations. We then show that transformations effected by feedforward networks always lie within this subset. This is the relatively easy part of the proof. The next subsection proves the harder part, namely that every transformation in this subset can indeed be effected by a feedforward network of depth equal to two.

Informally, a {\em causal transformation} is one whose
current output depends only on its past input (and not
current or future input). Abstractly, it is convenient to define a causal transformation as one that, given two different inputs that are identical until a certain point in time, also have their outputs, according to the transformation, be identical up to (at least) the same point.

\begin{definition}[Causal Transformation]
A transformation ${\cal T}:{\cal F}_m\rightarrow{\cal S}$ is said to be {\em
  causal} if, for every $\chi_1, \chi_2 \in {\cal F}_m$, with $\Xi_{(t,
    \infty)}\chi_1 =\Xi_{(t, \infty)}\chi_2$, for some $t\in \mathbb{R}$,
    we have $\Xi_{[t, \infty)} {\cal T}(\chi_1) =\Xi_{[t, \infty)}{\cal T}(\chi_2)$.
\end{definition}

As in signals and systems theory, a {\em time-invariant transformation} is one which always transforms the time-shifted version of an input, to a time-shifted version of its corresponding output. To keep the definition sound, we also need to ensure that the time-shifted input, in fact, also satisfies the Flush criterion.

\begin{definition}[Time-Invariant Transformation]
A transformation ${\cal T}:{\cal F}_m\rightarrow{\cal S}$ is said to be {\em time-invariant} if, for every $\chi \in {\cal F}_m$ and every $t\in \mathbb{R}$ with $\sigma_t(\chi) \in {\cal F}_m$, we have ${\cal T}(\sigma_t(\chi))=\sigma_t({\cal T}(\chi))$.
\end{definition}

A {\em resettable transformation} is one for which there
exists a positive real number $W$, so that an input gap of the form $(t, t+W)$
``resets'' it, i.e. output beyond $t$ is independent of input
received before it. Again, abstractly, it becomes convenient to say that the output in this case is identical to that produced by an input which has no spikes before $t$, but is identical to the present input thereafter.

\begin{definition}[$W$-Resettable Transformation]
For $W \in \mathbb{R}^+$, a transformation ${\cal T}:{\cal
  F}_m\rightarrow{\cal S}$ is said to be {\em $W$-resettable} if, for
every $\chi \in {\cal F}_m$ which has a gap in the interval $(t,
  t+W)$, for some $t\in \mathbb{R}$, we have $\Xi_{(-\infty, t]} {\cal T}(\chi) = {\cal T}(\Xi_{(-\infty, t]} \chi)$.
\end{definition}

\begin{definition}[Resettable Transformation]
A transformation ${\cal T}:{\cal F}_m\rightarrow{\cal S}$ is said to be {\em resettable} if, there exists a $W \in \mathbb{R}^+$, so that it is $W$-resettable.
\end{definition}

Next, we prove that every transformation that can be effected by a feedforward network is causal, time-invariant and resettable, in the context of our neuron model and its assumptions.

\begin{lemma}\label{lemma1}
If ~${\cal T}:{\cal F}_m\rightarrow{\cal S}$ can be effected by a feedforward network, then ${\cal T}(\cdot)$ is causal, time-invariant and resettable.
\end{lemma}

\begin{proof}[Proof sketch]
If ~${\cal T}:{\cal F}_m\rightarrow{\cal S}$ can be effected by a single neuron it is relatively straightforward to verify that ${\cal T}(\cdot)$ is causal, time-invariant and resettable. That it is causal and time-invariant follows from the fact that the $P(\cdot)$ function of the neuron only ``looks'' at the recent past and not the present or the future to determine membrane potential. That ${\cal T}(\cdot)$  is resettable follows from Axiom (3) of the neuron and the Gap Lemma. For a feedforward network, the proof proceeds by mathematical induction on the depth of the network. A full proof is provided in Appendix B.
\end{proof}

\subsection{Construction of a depth two feedforward network for every causal, time-invariant and resettable transformation}

In this subsection, we prove the following lemma.

\begin{lemma}
If ~${\cal T}:{\cal F}_m\rightarrow{\cal S}$ is causal, time-invariant and resettable, then it can be effected by a feedforward network of depth two.
\end{lemma}
\addtocounter{lemma}{-1}

\noindent
Before diving into the proofs, we offer some intuition.

Suppose we had a transformation ${\cal T}:{\cal F}_m\rightarrow{\cal S}$ which is causal, time-invariant and resettable. For the moment, pretend it satisfies the following property: There exist constant-sized input and output ``windows'' so that, for every input spike-train ensemble satisfying a flush criterion, just given knowledge of spikes in those windows of past input and output, one can unambiguously determine, at any point in time, if the transformation prescribes an output spike or not. Intuitively, it seems reasonable that such a transformation can be effected by a single neuron\footnote{Strictly speaking, it turns out that this is not true; axiom 2 may be violated.} by setting the $\Upsilon$ and $\rho$ of the neuron to the sizes of the input and output windows mentioned above.

Of course, one easily sees that not every transformation that is causal, time-invariant and resettable satisfies the aforementioned property. That is, there could exist two different input instances, whose past inputs and outputs are identical in the aforementioned windows at some points in time; yet in one instance, the transformation prescribes an output spike, whereas it prescribes none in the other. Indeed, the two input instances must differ at some point in the past, for otherwise the transformation would not be causal. Therefore, in such a situation, it is natural to ask if a single ``intermediate'' neuron can ``break the tie''. That is, if two input instances differ at some point in the past, the output of the intermediate neuron since then, in any interval of time of length $U$, must be different in either case, where $U$ is a fixed constant. This is so that a neuron receiving input from the intermediate neuron can {\em disambiguate} the two inputs, were an output spike demanded for one input but not the other. Unfortunately, this exact property cannot be achieved by any single ``tie-breaker'' neuron because every transformation induced by a neuron is resettable. In other words, the problem is that, suppose two input instances differ at a certain point in time; however, since then, both have had an arbitrarily large input gap. The input gap serves to ``erase memory'' in any network that received it and therefore it cannot disambiguate two inputs beyond this gap. Now, fortunately, it does not have to, since this gap also causes a ``reset'' in the transformation (which is resettable). That is, if such an arbitrarily large gap were present in the input, the transformation would not afterward demand an output spike in one case and no output spike in another. This is because it is $W$-resettable and therefore cannot make such demands, for input gaps\footnote{which we call a ``reset gap'' from now on, for the sake of exposition.} larger than $W$. Thus, we can make do with a slightly weaker condition; that the intermediate neuron is only guaranteed to break the tie, when it is required to do so. That is, suppose there are two input instances, whose outputs according to ${\cal T}:{\cal F}_m\rightarrow{\cal S}$ are different at certain points in time. Then, the corresponding inputs are different too at some point in the past with no reset gaps in the intervening time and therefore the intermediate neuron ought to break the tie. Additionally, for technical reasons that will become clear later, we stipulate that the outputs of the intermediate neuron in the preceding $U$ milliseconds  are guaranteed to be different, only if the inputs themselves in the past $U$ milliseconds are not different.

\begin{figure}
\begin{center}
  \includegraphics[width=2.5cm,height=3.2cm]{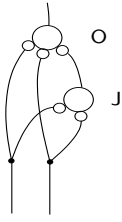}
\caption{The network architecture for (order two) feedforward networks of depth two equipped with model neurons described in Section \ref{sec:Model} that can effect any causal, time-invariant and resettable transformation.}
\label{fig:n}
\end{center}
\end{figure}

The network we have in mind is illustrated in Figure \ref{fig:n}, for $m=2$. In the following proposition, we prove that if the intermediate neuron satisfies the ``tie-breaker'' condition alluded to above, then there exists an output neuron, so that the network effects the transformation in question. Thereafter, in the subsequent proposition, we provide a construction for the intermediate neuron that satisfies this condition. By way of notation, recall that $\Xi_0(\cdot)$ is shorthand for $\Xi_{[0,0]}(\cdot)$

\begin{proposition}\label{depth2prop1}
Let ~${\cal T}:{\cal F}_m\rightarrow{\cal S}$ be causal, time-invariant and resettable. Let $\mathsf{J}$ be a neuron with ${\cal T}_{\mathsf{J}}:{\cal F}_m\rightarrow{\cal S}$, so that for each $\chi \in {\cal F}_m$, ${\cal T}_{\mathsf{J}}(\chi)$ is consistent with $\chi$ with respect to $\mathsf{J}$. Further, suppose there exists a $U \in \mathbb{R}^+$ so that for all $t_1, t_2 \in \mathbb{R}$ and $\chi_1, \chi_2 \in {\cal F}_m$ with $\Xi_0 \sigma_{t_1}( {\cal T}(\chi_1))\neq \Xi_0 \sigma_{t_2}( {\cal T}(\chi_2))$, we have $\Xi_{(0,U)} (\sigma_{t_1}({\cal T}_{\mathsf{J}}(\chi_1)\sqcup \chi_1))\neq \Xi_{(0,U)} (\sigma_{t_2}({\cal T}_{\mathsf{J}}(\chi_2)\sqcup \chi_2))$.

Then, there exists a neuron $\mathsf{O}$, so that for every $\chi \in {\cal F}_m$, ${\cal T}(\chi)$ is consistent with ${\cal T}_{\mathsf{J}}(\chi) \sqcup \chi$ with respect to $\mathsf{O}$.

\end{proposition} 

\begin{proof}[Proof sketch]
The straightforward way for the neuron $\mathsf{O}$ to effect ${\cal T}(\cdot)$ is to determine the points of time wherein an output spike is prescribed and set its membrane potential function to hit threshold at those instances. Since the neuron $\mathsf{J}$ essentially ``disambiguates'' the input, this assignment can be done without conflict. However, we also need to show that doing this does not violate any of the three axioms of our abstract model, for the neuron $\mathsf{O}$.  Axiom (1) follows easily from the fact that the co-domain of ${\cal T}(\cdot)$ is ${\cal S}$. Axiom (3) takes some work to show and uses the fact that ${\cal T}(\cdot)$ is causal, time-invariant and resettable. Axiom (2), on the other hand, presents some subtleties. Now, in addition to setting membrane potential to threshold at the aforementioned points, in order to satisfy Axiom (2), we would also need to set it to hit threshold, when the input window has the same pattern and the output window is empty instead. However, with this assignment, we need to then show that no spurious spikes are generated. This takes a little work and again uses the ``tie-breaker'' condition of the intermediate neuron $\mathsf{J}$. The full proof is available in Appendix B.
\end{proof}

\noindent

The next proposition shows that one can always construct an intermediate neuron that satisfies the said ``tie-breaker'' condition.

\begin{proposition}\label{depth2prop2}
Let ${\cal T}:{\cal F}_m\rightarrow{\cal S}$ be causal, time-invariant and resettable. Then there exists a neuron $\mathsf{J}$ and $U \in \mathbb{R}^+$ so that for all $t_1, t_2 \in \mathbb{R}$ and $\chi_1, \chi_2 \in {\cal F}_m$ with $ \Xi_{0} \sigma_{t_1}({\cal T}(\chi_1))\neq \Xi_{0} \sigma_{t_2}({\cal T}(\chi_2))$, we have $\Xi_{(0,U)} (\sigma_{t_1}({\cal T}_{\mathsf{J}}(\chi_1)\sqcup \chi_1))\neq \Xi_{(0,U)} (\sigma_{t_2}({\cal T}_{\mathsf{J}}(\chi_2)\sqcup \chi_2))$, where ${\cal T}_{\mathsf{J}}:{\cal F}_m\rightarrow{\cal S}$ is such that for each $\chi \in {\cal F}_m$, ${\cal T}_{\mathsf{J}}(\chi)$ is consistent with $\chi$ with respect to $\mathsf{J}$. 

\end{proposition}

\begin{proof}[Proof idea]
The basic idea is to ``encode'', in the time difference of two successive output spikes, the positions of all the input spikes that have occurred since the last input gap of the form $(t, t+W)$, where ${\cal T}(\cdot)$ is $W$-resettable. Such pairs of output spikes are produced once every $p$ milliseconds, with the time difference within each pair being a function of the time difference within the previous pair and the input spikes encountered since. Intuitively, it is convenient to think of this encoding as one from which we can ``reconstruct'' the entire past input spike-train ensemble after the last reset gap in the input. We first describe the encoding function for the case of a single input spike-train after which we indicate how it can be generalized.

So, suppose the time difference of the successive spikes output by $\mathsf{J}$ lies in the interval $[0, 1)$. Define the encoding function as $\varepsilon_0:[0,1)\times {\bar{\cal S}_{(0, p]}} \rightarrow [0,1)$, that takes in the old encoding and the input spikes in the past $p$ milliseconds to produce the new encoding, which is output by $\mathsf{J}$ as the time difference between a new pair of spikes. The number $p$ is chosen to be such that there are at most $8$ spikes in any interval of the form $(t, t+p]$. We now describe how $\varepsilon_0(e, {\vec x})$ is computed, given $e \in [0,1)$ and ${\vec x} = \langle x^1, x^2, \ldots, x^k \rangle$, such that each spike time in ${\vec x}$ lies in the interval $(0, p]$. Let $e$ have a decimal expansion\footnote{Whenever we say decimal expansion, we forbid decimal expansions with an infinite number of successive $9$s. With this restriction, each real number has a unique decimal expansion.}, so that $e=0.c_1s_1c_2s_2c_3s_3\cdots$. Accordingly, let $c=0.c_1c_2c_3\cdots$ and $s=0.s_1s_2s_3\cdots$. $c$ is a real number that encodes the number of spikes in each interval of length $p$ encountered, since the last reset. Since each interval of length $p$ has between $0$ and $8$ spikes, the digit $9$ is used as a ``termination symbol''. So, for example, suppose there have been $4$ intervals of length $p$, since the last reset with $5, 0, 8 $ and $2$ spikes apiece respectively, then $c=0.8059$ and $c'=0.28059$, where $c'$ is the ``updated'' value of $c$. Likewise, $s$ is a real number that stores the positions of all input spikes encountered since the last reset. Let each spike time be of the form $x^i=0.x^i_1x^i_2x^i_3\cdots \times 10^q$, for appropriate $q$, whose value is fixed for a given $p$. Then the updated value of $s$ is $s'=0.x^1_1x^2_1\cdots x^k_1s_1x^1_2x^2_2\cdots x^k_2s_2\cdots$. Suppose the $c'$ and $s'$ obtained above were of the form $c'=0.c'_1c'_2c'_3\cdots$ and $s'=0.s'_1s'_2s'_3\cdots$, then $\varepsilon_0(e, {\vec x})=0.c'_1s'_1c'_2s'_2\cdots$. Observe that the decimal expansion constructed by $\varepsilon_0(e, {\vec x})$ cannot have infinitely many successive $9$s, for $c'$ has only a finite number of non-zero digits. Suppose the input were a spike-train ensemble of order $m$, then for each spike-train an encoding would be computed as above and in the final step, the $m$ real numbers obtained would be interleaved together, so as to produce the encoding.

\addtocounter{footnote}{-3}
\begin{figure}
\begin{center}
  \includegraphics[width=8.7cm]{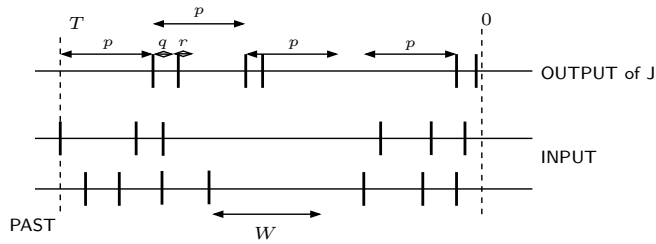}
\caption{This figure illustrates the operation of the intermediate neuron ~$\mathsf{J}$. Suppose $\chi \in {\cal F}_m$ is an input spike-train. Let its oldest spike be $T$ milliseconds ago. Then $\mathsf{J}$ produces a spike at time\protect\footnotemark $T-p$ and at every $T-kp$, for $k\in \mathbb{Z}^+$, unless in the previous $p$ milliseconds to when it is to spike, there is a gap\protect\footnotemark of the form $(t, t+W)$. For the sake of exposition, let's call these the ``clock'' spikes. Now, suppose there is a gap of the form $(t, t+W)$ in the input and there is an input spike at time $t$, then the neuron spikes at time $t-p$ and every $p$ milliseconds thereafter subject to the same ``rules'' as above. These clock spikes are followed by ``encoding'' spikes, which occur at least $q$ milliseconds after the clock spike, but less than $q+r$ milliseconds after, where $q$ is greater than the absolute refractory period $\alpha$. As expected, the position of the current encoding spike is a function of the time difference between the previous encoding and clock spikes\protect\footnotemark and the positions of the input spikes in the $p$ milliseconds before the current clock spike. The output of the encoding function is, in effect, appropriately scaled to ``fit'' in this interval of length $r$; the details are available in the proof.}
\label{fig:o}
\end{center}
\end{figure}

\addtocounter{footnote}{-3}
\stepcounter{footnote}\footnotetext{i.e. $p$ milliseconds after time instant ~$T$.}
\stepcounter{footnote}\footnotetext{We set $W>p$ to force a spike at $T-p$.}
\stepcounter{footnote}\footnotetext{unless the present clock spike is the first after a reset gap in the input.}

Given knowledge of the encoding function, Figure \ref{fig:o} briefly describes how $\mathsf{J}$ works. The claim then is that if two input spike-train ensembles are different at some point with no intervening ``reset'' gaps, then the output of $\mathsf{J}$ in the past $U$ milliseconds, where $U=p+q+r$ will be different. Intuitively, this is because the difference between the latest encoding and clock spike in each case would be different, as they encode different ``histories'' of input spikes. The exception is if the input spike-train ensembles differed only in the past $U$ milliseconds. In this case, the difference is communicated to $\mathsf{O}$ directly by $\chi$.

Finally, we ought to remark that the above is just an informal description that glosses over several technical details contained in the full proof, which is available in Appendix B.
\end{proof}

\noindent
The preceding two propositions thus imply Lemma~\ref{lemma2} which together with Lemma~\ref{lemma1} implies Theorem~\ref{theorem3}.

\begin{lemma}\label{lemma2}
If ~${\cal T}:{\cal F}_m\rightarrow{\cal S}$ is causal, time-invariant and resettable, then it can be effected by a feedforward network of depth two.
\end{lemma}

\begin{theorem}\label{theorem3}
If ~${\cal T}:{\cal F}_m\rightarrow{\cal S}$ can be effected by a feedforward network, then it can be effected by a feedforward network of depth two.
\end{theorem}
\begin{corollary}
The set of all feedforward networks is not more complex than the set of feedforward networks of depth equal to two.
\end{corollary}

Incidentally, Lemma \ref{lemma1} and \ref{lemma2} also lead to a full characterization of the class of transformations effected by all feedforward networks equipped with neurons obeying the abstract model of Section~\ref{sec:Model}. This is formalized in the next theorem.

\begin{theorem}
A transformation ${\cal T}:{\cal F}_m\rightarrow{\cal S}$ can be effected by a feedforward network if and only if it is causal, time-invariant and resettable.
\end{theorem}

\subsubsection*{Directions for further constraining the present model}

The results of this section imply that we need to add new properties to further constrain our model neurons, in order for complexity results involving feedforward networks of depth two to be manifested. There are a number of directions that one could take. One is that spike-times in the present model are real numbers. When stochastic variability in neurons is taken into account, this assumption is no longer true. Also, we did not assume that the membrane potential changes smoothly with time, which would be a reasonable assumption to add. And, finally, an assumption consistent with Dale's principle, that each neuron has either an excitatory effect on all its postsynaptic neurons or an inhibitory effect might also help in this direction.

\section{Discussion}\label{sec:disc}

There has been some debate about how useful data from the connectome
projects might be in advancing a mechanistic understanding of
computation occurring in the circuits of the brain. One of the main
type of arguments that has been made against their utility is that,
since these projects only\footnote{This in itself is a formidable
  problem and one that is taking heroic effort.} seek to ascertain the
wiring diagram, without giving us detailed physiological information,
it is not clear what we might learn from this data alone, especially
for networks whose high-level function is not known. While it is
acknowledged that network architecture places constraints on what
a network can compute \citep{kleinfeld2011large, denk}, the nature and
scope of these constraints have remained poorly understood. Our goal with
this work was in asking, on one hand, if we can deduce non-trivial
examples of computations that a network {\em could not} be doing,
given just the knowledge of its architecture and assuming that the
neurons obey some elementary properties. On the other hand, we asked
if there are fundamental limits to what can be said, given just this
information. We examined this question for the case of feedforward
networks equipped with neurons that obeyed a deterministic spiking
neuron model. We first set the stage by creating a mathematical
framework in which this question could be precisely posed. Crucially,
we needed to make precise what computation exactly meant in this
context. This took a fair bit of work and led us to the view of
feedforward networks as spike-train to spike-train transformations
under biologically-relevant spiking regimes. After setting up
necessary definitions, we then showed some examples of transformations
that networks of specific architectures {\em cannot} effect, that
other networks can. First of all, we showed\footnote{See Figure~\ref{fig:k} and the second paragraph of Section~\ref{sec:compl_results}.} that there exist spike-train to
spike-train transformations that no feedforward network could
effect. Next, we showed a transformation that no single neuron could
effect but a network consisting of two neurons could. After this, we
proved a result which shows that a class of architectures that share a
certain structural property also share their inability to effect
a particular class of transformations. Notably, while this class of
architectures has networks with arbitrarily many neurons, we showed a
class of networks with just two neurons which could effect this class
of transformations. This suggests that network structure alone may
impose crucial constraints on computational ability. Finally, we
demonstrated that the small number of properties assumed for our model
neurons can only take us so far. We proved that without making further
assumptions about our model neurons, we couldn't discern such examples
for the set of all feedforward networks of depth two.

While there is more to neuronal networks than just their
wiring diagram, what our theory suggests is that the wiring diagram
could impose crucial constraints on the computational ability of
networks, in some cases. On the other hand, there seem to be classes
of networks for which a more elaborate knowledge of single neuron
properties may be necessary, before we can determine restrictions on
their computational ability. While technical issues in electron microscopy \citep{denk} have so far stood in the way of mapping, for example, distributions of ion-channels and neurotransmitter and neuromodulator receptors in neurons, it is conceivable that such hurdles may be overcome in future.
If successful, these or other advances in conjunction with the wiring diagram could provide useful
information to help us tease out pertinent constraints on the
computational capabilities of these networks.

In this work, as a first step, we have aimed to demonstrate specific {\em examples} of computations that a network cannot accomplish, given its architecture. The more ambitious goal would be the ability to have an exact characterization of the set of {\em all} computations that a given neural circuit cannot perform, given knowledge of its architecture, to the extent that a given incomplete knowledge of the physiological properties of its neurons will allow. This is not necessarily a goal that is out of reach. Even in the present work, we have obtained such an exact characterization\footnote{This characterization is a consequence of Theorem 4. In particular, it is the set of all transformations that are {\em not} causal, time-invariant or resettable.} of the set of all computations that the set of feedforward networks cannot accomplish, given the set of properties that our model neurons are presently assumed to obey. Therefore, in principle, there seems to be no reason why we may not be able to do likewise for specific network architectures.

\subsubsection*{Acknowledgements}
This work was supported, in part, by a National Science Foundation grant (NSF IIS-0902230) to A.B.

\section*{Appendix A: Relationship of the abstract neuron model to some widely-used neuron models }
Here, we demonstrate that the properties that our abstract model of the neuron is contingent on are satisfied, up to arbitrary accuracy, by several widely-used neuron models such as the Leaky Integrate-and-Fire Model and Spike Response Model.

\subsection*{Leaky Integrate-and-Fire Model}
\noindent
Consider the standard form of the Leaky Integrate-and-Fire Model:

\begin{equation}
\tau_m \frac{du}{dt} = -u(t)+R I(t)
\end{equation}

where $\tau_m =RC$. When $u(t^{(f)})= v$, the neuron fires a spike and the reset is given by $u(t^{(f)}+\Delta)=u_r$, where $v$ is the threshold and $\Delta$ is the absolute refractory period. Suppose an output spike has occurred at time ${\hat t} - \Delta$, the above differential equation has the following solution:

\begin{equation}
u(t)=u_r \exp(-\frac{t-{\hat t}}{\tau_m}) + \frac{1}{C} \int_{0}^{t-{\hat t}} \exp(-\frac{s}{\tau_m}) I(t-s) ds
\end{equation}

Suppose $I(t)=\Sigma_j w_{j} \Sigma_i \alpha(t-t_j^{(i)})$ and $\alpha(\cdot)$ had a finite support. Then, it is clear from the above expression that the contribution of the previous output spike fired by the present neuron as well as the contribution of input spikes from presynaptic neurons decays exponentially with time. Therefore, one can compute the membrane potential to arbitrary accuracy by choosing input and output ``windows'' of appropriate size so that $u(\cdot)$ is a function only of input spikes and output spikes in those windows. It is easy to verify that the all the axioms of our model are satisfied: Clearly, the model above has an absolute refractory period, a past output spike has an inhibitory effect on membrane potential, and upon receiving no input and output spikes in the said windows, it settles to resting potential. Thus, an instantiation of our abstract model can simulate a Leaky Integrate-and-Fire Model to arbitrary accuracy.

\subsection*{Spike Response Model}
\noindent
Consider now the standard form of the Spike Response Model\citep{gerst}.

In the absence of spikes, the membrane potential $u(\cdot)$ is set to the value $u_{r}=0$. Otherwise, the membrane potential is given by

\begin{equation}
u(t) = \eta(t-{\hat t_i}) + \Sigma_j~ w_{j}~ \Sigma_i~ \epsilon_{ij}(t-{\hat t_i}, t-t_j^{(i)})
\end{equation}

where $\eta(\cdot)$ describes the after-hyperpolarization after an output spike at ${\hat t_i}$ and $\epsilon_{ij}(\cdot)$ describes the response to incoming spikes $t_j^{(i)}$, which are the spikes fired by presynaptic neuron $j$ with $w_{j}$ being synaptic weights. $\eta(\cdot)$, is set to a sufficiently low value for $\Delta$ milliseconds after an output spike so as not to cause another spike, where $\Delta$ is the absolute refractory period. The functions $\eta(\cdot)$ and $\epsilon_{ij}(\cdot)$ typically decay exponentially with time and therefore, as before, one can compute the membrane potential to arbitrary accuracy by choosing input and output ``windows'' of appropriate size so that the $u(\cdot)$ is a function only of input spikes and output spikes in those windows. Likewise, it is easy to verify that the all the axioms of our model are satisfied: Clearly, the model above has an absolute refractory period, a past output spike has an inhibitory effect on membrane potential, and upon receiving no input and output spikes in the said windows, it settles to resting potential. Thus, it is straightforward to verify that an instantiation of our abstract model can simulate a Spike Response Model to arbitrary accuracy.

\newpage

\section*{Appendix B: Proofs and Technical Remarks}
\setcounter{lemma}{0}
\setcounter{proposition}{0}
\setcounter{corollary}{0}

\subsection*{Technical Remarks from Section~\ref{sec:counterexample}}
It might be argued that the input spike-train to a neuron cannot possibly be
infinitely long, since every neuron begins existence at a certain
point in time. However, this begs the question whether the neuron was at
the resting potential when the first input spikes
arrived\footnote{Note that our axiomatic definition of a neuron does not
  address this question.}. An assumption to this effect would be
significant, particularly if the current membrane potential depended
on it. It is easy to construct an example along the lines of the
example described in Figure \ref{fig:a}, where
the current membrane potential is different depending on whether this
assumption is made or not. Assuming infinitely long input spike-train ensembles, on the other hand, obviates the need to make any such
assumption. We retain this viewpoint for the rest of the paper with
the understanding that the alternative viewpoint discussed at the
beginning of this paragraph can
also be expounded along similar lines.

\subsection*{Proofs from Section~\ref{sec:GapLemma}}
\begin{proof}[Proof of Gap Lemma]
Since, in each ${\vec x_0}$ consistent with
$\chi$, with respect to ${\mathsf N}$, the interval $(t+2 \rho, t+3 \rho)$ of ${\vec x_0}$ and the
$(t+\Upsilon+\rho, t+\Upsilon+2\rho)$ of $\chi$ are arbitrary, the
sequence of spikes present in the interval $(t+\rho, t+2\rho)$ of ${\vec
  x_0}$ could be arbitrary. However, $\chi^*$ and $\chi$ are identical
in $(t, t+\rho+\Upsilon)$. Thus, it follows from Axiom 2 in the
formal definition of a neuron that for
every $t' \in (t,t+\rho)$, $P( \Xi_{(0,\Upsilon)}(\sigma_{t'}(\chi)),
\Xi_{(0,\rho)}(\sigma_{t'}({\vec x_0})))$ is at most the value of $P(
\Xi_{(0,\Upsilon)}(\sigma_{t'}(\chi^*)),
\Xi_{(0,\rho)}(\sigma_{t'}({\vec x_0^*})))$ , because
$\Xi_{(0,\rho)}(\sigma_{t'}({\vec x_0^*}))$ is ${\vec \phi}$, i.e. empty. Since
$P(\Xi_{(0,\Upsilon)}(\sigma_{t'}(\chi^*)),
\Xi_{(0,\rho)}(\sigma_{t'}({\vec x_0^*})))$ is less than $\tau$ for
every $t'\in (t,t+\rho)$, \\$P( \Xi_{(0,\Upsilon)}(\sigma_{t'}(\chi)),
\Xi_{(0,\rho)}(\sigma_{t'}({\vec x_0})))$ is less than $\tau$ in
the same interval, as well. Therefore, ${\vec x_0}$ has no
spikes in $(t, t+\rho)$.

That $2\rho$ is the smallest possible gap length in ${\vec
  x_0^*}$ for this to hold, follows from the counterexample in Figure \ref{fig:a},
where the present conclusion did not hold, when ${\vec x_0^*}$ had gaps of
length $2 \rho -\delta$, for arbitrarily small $\delta>0$.

\end{proof}

\begin{proof}[Proof of Corollary~\ref{corrgap}]
(\ref{c11}) is immediate from the Gap Lemma, when we set $\chi=\chi^*$.

For (\ref{c12}), the proof is by strong induction on the number of
spikes since $t$. Let ${\vec x_0}$ be an arbitrary spike-train that is
consistent with $\chi^*$, with respect to ${\mathsf N}$.  Notice that
from (\ref{c11}) we have that ${\vec x_0}$ is identical to
${\vec x_0}^*$ in $(t, t+\rho)$. The base case is to show that both
${\vec x_0}^*$ and ${\vec x_0}$ have their first spike since $t$ at
the same time. Assume, without loss of generality, that the first spike
of ${\vec x_0}$ at $t_1 \leq t$, is no later than the first spike of
${\vec x_0}^*$.  We have $P(
\Xi_{(0,\Upsilon)}(\sigma_{t_1}(\chi^*)),
\Xi_{(0,\rho)}(\sigma_{t_1}({\vec x_0^*})))=P(
\Xi_{(0,\Upsilon)}(\sigma_{t_1}(\chi^*)),
\Xi_{(0,\rho)}(\sigma_{t_1}({\vec x_0})))$ since
$\Xi_{(0,\rho)}(\sigma_{t_1}({\vec
  x_0^*}))=\Xi_{(0,\rho)}(\sigma_{t_1}({\vec x_0}))={\vec
  \phi}$. Therefore ${\vec x_0^*}$ also has its first spike since $t$ at
$t_1$. Let the induction hypothesis be that both ${\vec x_0^*}$ and
${\vec x_0}$ have their first $k$ spikes since $t$ at the same
times. We show that this implies that the $(k+1)^{th}$ spike in each
spike-train is also at the same time instant.  Assume, without loss of
generality, that the $(k+1)^{th}$ spike since $t$ of ${\vec x_0}$ at
$t_{k+1}$, is no later than the $(k+1)^{th}$ spike since $t$ of ${\vec
  x_0}^*$. Now, $\Xi_{(0,\rho)}(\sigma_{t_{k+1}}({\vec
  x_0^*}))$ is identical to $\Xi_{(0,\rho)}(\sigma_{t_{k+1}}({\vec x_0}))$ from the
induction hypothesis since $(t+\rho) - t_{k+1} \geq \rho$. Thus, $P(
\Xi_{(0,\Upsilon)}(\sigma_{t_{k+1}}(\chi^*)),
\Xi_{(0,\rho)}(\sigma_{t_{k+1}}({\vec x_0^*})))=P(
\Xi_{(0,\Upsilon)}(\sigma_{t_{k+1}}(\chi^*)),
\Xi_{(0,\rho)}(\sigma_{t_{k+1}}({\vec x_0})))$ and therefore ${\vec x_0^*}$ also has its $(k+1)^{th}$  spike at
$t_{k+1}$. This completes the proof of (\ref{c12}).

(\ref{c13}) follows from the Gap Lemma and (\ref{c12}).

\end{proof}

\begin{proposition}\label{prop1}
Let $\chi$ be a spike-train ensemble that satisfies a T-Gap criterion for
a neuron ${\mathsf N}\langle\alpha,
\Upsilon, \rho, \tau, \lambda, m, P:{\bar{\cal S}_{(0, \Upsilon)}}^m 
\times \bar{\cal S}_{(0, \rho)} \rightarrow [\lambda, \tau]\rangle$, where $T \in \mathbb{R}^+$. Then, there is exactly
one spike-train ${\vec x_0}$, such that ${\vec x_0}$ is consistent with $\chi$, with respect to ${\mathsf N}$.
\end{proposition}

\begin{proof}[Proof of Proposition~\ref{prop1}]
Since $\chi$ satisfies a $T$-Gap criterion, there exists a spike-train
${\vec x_0}$ with at least one gap of length $2\rho$ in every interval
of time of length $T-\Upsilon +2\rho$, so that ${\vec x_0}$ is
consistent with $\chi$ with respect to ${\mathsf N}$. For the sake of
contradiction, assume that there exists another spike-train ${\vec
  x_0}'$, not identical to ${\vec x_0}$, which is consistent
with $\chi$, with respect to ${\mathsf N}$. Let $t'$ be the time at
which one spike-train has a spike but another doesn't. Let $t>t'$ be
such that ${\vec x_0}$ has a gap in the interval $(t, t+2\rho)$. By
Corollary \ref{corrgap} to the Gap Lemma, it follows that ${\vec
  x_0}'$ is identical to ${\vec x_0}$ after time instant $t+\rho$. This contradicts the hypothesis that ${\vec x_0}'$ is different from ${\vec x_0}$ at t'.

\end{proof}

\addtocounter{lemma}{1}

\begin{lemma}\label{lemma2:gap}
Consider a feedforward network ${\cal N}$. Let $\chi$ satisfy a $T$-Gap criterion for ${\cal N}$, where $T \in \mathbb{R}^+$. Then the output neuron of ${\cal N}$ produces a unique output spike-train when ${\cal N}$ receives $\chi$ as input. Furthermore, the membrane potential of the output neuron at any time instant depends on at most the past $T$ milliseconds of input in $\chi$.
\end{lemma}

\begin{proof}[Proof of Lemma~\ref{lemma2:gap}]
We prove that the output of the network is unique by strong  induction on depth. Let $N_i$, for $1\leq i\leq d$, be the set of neurons in ${\cal N}$ of depth $i$. Each neuron ${\mathsf N}\in N_1$ receives all inputs from spike-trains in $\chi$. Since, ${\mathsf N}$ satisfies a Gap criterion with those input spike-trains, its output is unique. The induction hypothesis then is that for all $i \leq k<d$, each neuron ${\mathsf N} \in N_i$ produces a unique output spike-train when ${\cal N}$ is driven by $\chi$. Consider arbitrary ${\mathsf N}' \in N_{k+1}$. It is clear that all inputs to ${\mathsf N}'$ are from spike-trains from $\chi$ or neurons in $\bigcup_{i=1}^k N_i$, for otherwise the depth of ${\mathsf N}'$ would be greater than $k+1$. Since, all its inputs are unique by the induction hypothesis and they satisfy a Gap criterion for ${\mathsf N}'$, its output is also unique.

Next, we show that the membrane potential of the output neuron at any time instant depends on at most the past $T$ milliseconds of input in $\chi$. Since the output neuron satisfies a $(\frac{T}{d})$-Gap
Criterion, its membrane potential at any point depends on at most the
past $(\frac{T}{d})$ milliseconds of the inputs it receives (some of which
may be output spike-trains of other neurons). Consider one such
``penultimate layer'' neuron. Again, its output membrane potential at any
time instant, likewise, depends on its inputs in the past $(\frac{T}{d})$
milliseconds. Therefore, the current potential of the output neuron is
dependent on the input received by the penultimate layer neuron in at most the
past $(\frac{2T}{d})$ milliseconds. Similar arguments can be put forth
until, for each path, one reaches a neuron, all of whose inputs do not
come from other neurons. Since the longest such path is of length $d$,
it is straightforward to verify that the membrane potential of the output neuron depends on at most $T$
milliseconds of past input in $\chi$.

\end{proof}

\subsection*{Proofs from Section~\ref{sec:Flush}}

\begin{figure}
\begin{center}
  \includegraphics[width=8.6cm]{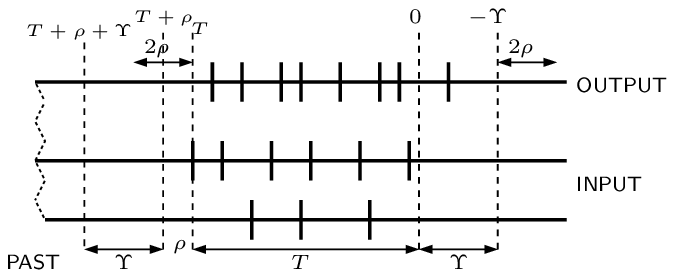}
\caption{Illustration showing that an input spike-train ensemble satisfying a Flush Criterion also satisfies a Gap Criterion.}
\label{fig:f}
\end{center}
\end{figure}

\begin{lemma}\label{lemma3}
An input spike-train ensemble $\chi$~ for a neuron ${\mathsf N}\langle\alpha, \Upsilon, \rho, \tau,
\lambda, m, P:{\bar{\cal S}_{(0, \Upsilon)}}^m  \times \bar{\cal
  S}_{(0, \rho)} \rightarrow [\lambda, \tau]\rangle$ that satisfies a $T$-Flush
Criterion also satisfies a $(T+2\Upsilon+2\rho)$-Gap Criterion for
that neuron.
\end{lemma}

\begin{proof}[Proof of Lemma~\ref{lemma3}]
Figure \ref{fig:f} accompanies this proof. The neuron on being driven by $\chi$ cannot have output spikes outside
the interval $(-\Upsilon, T)$. This easily follows from Axiom 2 and 3
of the neuron because the neuron does not
have input spikes before time instant $T$ and in the interval
$(-\Upsilon, 0)$ and onwards. Now, to see that $\chi$ satisfies a
$(T+2\Upsilon+2\rho)$-Gap Criterion, recall that with a $T'$-Gap
Criterion, distance between any two gaps of length $2\rho$ on the
output spike-train is at most $T'-\Upsilon -2\rho$. With $\chi$, we
observe that the distance between any two $2\rho$ gaps on the output
spike-train is at most $T+\Upsilon$. Thus, $T'-\Upsilon -2\rho =
T+\Upsilon$, which gives us $T'=T+2\Upsilon+2\rho$. The result follows.

\end{proof}

\begin{lemma}\label{lemma4}
An input spike-train ensemble $\chi$~ for a feedforward network that
satisfies a $T$-Flush Criterion also satisfies a $(dT+d(d+1)\Upsilon
+2d\rho)$-Gap Criterion for that network, where $\Upsilon$, $\rho$
are upper bounds on the same parameters taken over all the neurons in
the network and $d$ is the depth of the network.
\end{lemma}

\begin{proof}[Proof of Lemma~\ref{lemma4}]
Following the proof of the previous lemma, we know that neurons that
receive all their inputs from $\chi$ have no output spikes outside the
interval $(-\Upsilon, T)$. Similarly, neurons that have depth 2 with
respect to the input vertices of the network have no output spikes
outside $(-2\Upsilon, T)$. Likewise, the output neuron, which has
depth $d$, has no output spikes outside $(-d\Upsilon, T)$. It follows
that the output neuron obeys a $(T+ (d+1)\Upsilon +2\rho)$-Gap
Criterion. Also, every other neuron obeys this criterion because the
distance between the $2\rho$ output gaps for every neuron is at most that of
the output neuron, since their depth is bounded from above by the
depth of the output neuron. Thus, from the definition of the Gap
criterion for feedforward networks, we have that $\chi$ satisfies a $(dT+d(d+1)\Upsilon
+2d\rho)$-Gap Criterion for the current network.

\end{proof}

\subsection*{Proofs from Section~\ref{sec:compl}}
\begin{proof}[Proof of Lemma~\ref{thm:GFeq}]
We prove the easy direction first.  If $\exists {\cal N}' \in \Sigma_2$ such that $\forall {\cal
  N} \in \Sigma_1,  {\cal T}_{{\cal N}'}|_{{\scriptscriptstyle{\cal F}_{m}}} \neq {\cal
  T}_{{\cal N}}|_{{\scriptscriptstyle{\cal F}_{m}}}$, then it follows that ${\cal T}_{{\cal
    N}'}|_{{\scriptscriptstyle{\cal G}_{12}}} \neq {\cal T}_{{\cal N}}|_{{\scriptscriptstyle{\cal G}_{12}}}$
because  ${\cal F}_m \subseteq {\cal G}_{\cal N}$.

For the other direction, let $\exists {\cal N}' \in \Sigma_2$ such
that $\forall {\cal N} \in \Sigma_1, {\cal T}_{{\cal N}'}|_{{\scriptscriptstyle{\cal G}_{12}}} \neq {\cal T}_{{\cal N}}|_{{\scriptscriptstyle{\cal G}_{12}}}$. We construct
${\cal F}'\subseteq {\cal F}_m $, so that ${\cal T}_{{\cal N}'}|_{{\scriptscriptstyle{{\cal F}'}}} \neq {\cal T}_{{\cal N}}|_{{\scriptscriptstyle{{\cal F}'}}}$. This immediately
implies ${\cal T}_{{\cal N}'}|_{{\scriptscriptstyle{\cal F}_{m}}} \neq {\cal T}_{{\cal
    N}}|_{{\scriptscriptstyle{\cal F}_{m}}}$. Consider arbitrary ${\cal N} \in
\Sigma_1$. From the hypothesis, we have ${\cal T}_{{\cal N}'}|_{{\scriptscriptstyle{\cal G}_{12}}} \neq {\cal T}_{{\cal N}}|_{{\scriptscriptstyle{\cal G}_{12}}}$. Therefore
$\exists \chi \in {\cal G}_{12}$ such that ${\cal T}_{{\cal
    N}'}|_{{\scriptscriptstyle{\cal G}_{12}}}(\chi) \neq {\cal T}_{{\cal N}}|_{{\scriptscriptstyle{\cal G}_{12}}}(\chi)$. Additionally, there exist $T_1, T_2 \in
\mathbb{R}^+$, so that $\chi$ satisfies a $T_1$-Gap Criterion for
${\cal N}$ and a $T_2$-Gap Criterion for ${\cal N}'$. Let $T=\max(T_1,
T_2)$. Let ${\cal T}_{{\cal N}'}|_{{\scriptscriptstyle{\cal G}_{12}}}(\chi) = {\vec x_0}'$
and $ {\cal T}_{{\cal N}}|_{{\scriptscriptstyle{\cal G}_{12}}}(\chi)= {\vec x_0} $ . Let
${\tilde{\cal F}}= \bigcup_{t \in \mathbb{R}}
\Xi_{(0,2T)}(\sigma_t(\chi))$. Note that each element of ${\tilde{\cal F}}$
satisfies a $2T$-Flush Criterion. The claim, then, is that ${\cal
  T}_{{\cal N}'}|_{{\scriptscriptstyle{\tilde{\cal F}}}} \neq {\cal T}_{{\cal
    N}}|_{{\scriptscriptstyle{\tilde{\cal F}}}}$. We have $ \Xi_{(0,T)}({\cal
  T}_{{\cal N}'}(\Xi_{(0,2T)}(\sigma_t(\chi)))) =
\Xi_{(0,T)}(\sigma_t({\vec x_0}'))$ and $ \Xi_{(0,T)}({\cal
  T}_{{\cal N}}(\Xi_{(0,2T)}(\sigma_t(\chi)))) =
\Xi_{(0,T)}(\sigma_t({\vec x_0}))$. This follows from the fact that
$\chi$ satisfies the $T$-Gap Criterion with both ${\cal N}$ and ${\cal
  N}'$ and therefore when  ${\cal N}$ and ${\cal
  N}'$ are driven by any segment of $\chi$ of length $2T$, the output
produced in the latter $T$ milliseconds of that interval agrees with ${\vec
  x_0}$ and ${\vec x_0}'$ respectively. Therefore, if ${\vec
  x_0} \neq {\vec x_0}'$, it is clear that there exists a $t$, so
that ${\cal T}_{{\cal N}'}(\Xi_{[0,2T]}(\sigma_t(\chi))) \neq {\cal
  T}_{{\cal N}}(\Xi_{[0,2T]}(\sigma_t(\chi)))$. ${\cal F}'$ is
obtained by taking the union of such ${\tilde{\cal F}}$ for every
${\cal N} \in \Sigma_1$. The result follows.

\end{proof}

\subsection*{Technical Remarks from Section~\ref{sec:compl_results}}
Some technical remarks concerning the mechanics of proving complexity results are stated below.

For two sets of feedforward networks, $\Sigma_1$ and $\Sigma_2$ with $\Sigma_1 \subseteq \Sigma_2$, in order to prove that $\Sigma_2$ is more complex than $\Sigma_1$, it is sufficient to show a transformation ${\cal T}: {\cal F}_{m} \rightarrow {\cal S}$ that no network present in $\Sigma_1$ can perform, while demonstrating a network in $\Sigma_2$ that can effect it. This involves constructing such a transformation, i.e. prescribing an output spike train for every element in ${\cal F}_{m}$. Recall that ${\cal F}_{m}$ consists of spike-train ensembles of order $m$, with the property that for each such ensemble there exists a positive real number $T$, so that the ensemble satisfies a $T$-Flush criterion. In practice, however, it usually suffices to prescribe output spike trains for a small subset\footnote{albeit typically one that contains, for each positive real number  $T$, at least one spike-train ensemble satisfying a $T$-Flush Criterion.} of elements of ${\cal F}_{m}$, and prove that no network in $\Sigma_1$ can map the input spike trains in that subset to their prescribed outputs. The second step would involve demonstrating a network in $\Sigma_2$ that maps this subset of ${\cal F}_{m}$ to the prescribed output, while mapping the rest of ${\cal F}_{m}$ to arbitrary output spike trains. Strictly speaking then, the transformation ${\cal T}: {\cal F}_{m} \rightarrow {\cal S}$ we prescribe comprises the mapping from ${\cal F}_{m}$ to output spike trains, as effected by {\em this} network in $\Sigma_2$. For convenience however, we shall refer to the mapping prescribed for some small subset of ${\cal F}_{m}$ as the prescribed transformation.

The next remark concerns timescales of the parameters $\Upsilon$ and $\rho$ of each neuron in the network and the timescale at which the transformation operates. Recall that the parameters $\Upsilon$ and $\rho$ correspond to the timescale at which the neuron integrates inputs it receives and the relative refractory period respectively. It would be reasonable to expect that the values of these parameters lie within a certain range as constrained by physiology, although this range might be different for different local neuronal networks in the brain. Suppose we have an upper bound on the value of each such parameter. Then, when we prove a complexity result, there would exist a timescale $T$, which is a function of these upper bounds, such that there exists a transformation on this timescale that cannot be performed by any network with the said architecture, whose parameters are governed by these upper bounds. More precisely, there would exist a transformation that maps a set of inputs satisfying a $T$-Flush criterion to an output spike train that (provably) cannot be performed by any network with the architecture in question. When stating and proving a complexity result, however, for the sake of succinctness, we do not explicitly state the relation between these bounds and the corresponding $T$. We simply let $\Upsilon$,  $\rho$ and $T$ remain unbounded. It is straightforward for the reader to derive a bound on $T$ as a function of bounds on $\Upsilon$ and $\rho$, as discussed.

The final remark is about our neuron model and the issue of what we can assume about the neurons when demonstrating that a certain network {\em can} effect a given transformation. Recall that our neuron model assumes that our neurons satisfy a small number of elementary properties but are otherwise unconstrained. This allowed our model to accomodate  a large variety of neuronal responses. This was convenient when faced with the task of showing that no network of a certain architecture could perform a given transformation, no matter what response properties its neurons have. However, when we wish to show that a certain transformation can be done by a specific network, some caution is in order. In this case, it is prudent to restrict ourselves to as simple a neuron model as possible, so that whether the neuronal responses employed are achievable by a real biological neuron, is not in question. In practice, we describe the neurons in the construction, so that they can certainly be effected by a highly-reduced neuron model such as the Spike Response Model SRM$_0$ \citep{gerst}.

\subsection*{Proofs from Section~\ref{sec:depth2}}

\begin{proof}[Proof of Lemma~\ref{lemma1}]
Let ${\cal N}$ be a network that effects ${\cal T}:{\cal F}_m\rightarrow{\cal S}$. 

\noindent
{\bf ${\cal T}(\cdot)$ is causal.} Consider arbitrary $\chi_1, \chi_2 \in {\cal F}_m$ with $\Xi_{(t,\infty)}\chi_1 =\Xi_{(t, \infty)}\chi_2$, for some $t\in \mathbb{R}$. We wish to show that $\Xi_{[t, \infty)} {\cal T}(\chi_1) =\Xi_{[t, \infty)}{\cal T}(\chi_2)$. Let $N_i$, for $1\leq i\leq d$, be the set of neurons in ${\cal N}$ of depth $i$, where $d$ is the depth of ${\cal N}$. Each neuron ${\mathsf N}\in N_1$ receives all its inputs from spike-trains in $\chi$. When the network receives $\chi_1$ and $\chi_2$ as input, suppose ${\mathsf N}$ receives $\chi'_1$ and $\chi'_2$ respectively as input. Also, clearly, $\Xi_{(t,\infty)}\chi'_1 =\Xi_{(t, \infty)}\chi'_2$. Let ${\vec x_1}'$ and ${\vec x_2}'$ be the output produced by ${\mathsf N}$ on receiving $\chi'_1$ and $\chi'_2$ respectively. Since  $\chi'_1, \chi'_2 \in {\cal F}_m$, there exists a $T \in \mathbb{R}^+$, so that $\Xi_{[T,\infty)}\chi'_1 = \Xi_{[T,\infty)}\chi'_2 ={\vec \phi}^{m'}$, where $m'$ is the number of inputs to $\mathsf{N}$. Therefore, by Axiom (3) of the neuron, we have $\Xi_{[T,\infty)}{\vec x_1}' = \Xi_{[T,\infty)}{\vec x_2}' ={\vec \phi}$.\/ Now, for all $t' \in
 \mathbb{R}$, $\Xi_{t'} {\vec x_j}'= \langle t' \rangle$ if and only if
 $P_{\mathsf{N}}(\Xi_{(0,\Upsilon_{\mathsf{N}})}(\sigma_{t'}(\chi'_j)), \Xi_{(0,\rho_{\mathsf{N}})}(\sigma_{t'}({\vec
   x_j}')) =\tau_{\mathsf{N}}$, for $j=1,2$. It is immediate that for $t' > t$, we have $\Xi_{(0,\Upsilon_{\mathsf{N}})}(\sigma_{t'}(\chi'_1))=\Xi_{(0,\Upsilon_{\mathsf{N}})}(\sigma_{t'}(\chi'_2))$. Now, by an induction argument on the spike number since $T$, it is straightforward to show that for all $t' > t$, $\Xi_{(0,\rho_{\mathsf{N}})}(\sigma_{t'}({\vec   x_1}'))=\Xi_{(0,\rho_{\mathsf{N}})}(\sigma_{t'}({\vec x_2}'))$. Thus, we have $\Xi_{[t,\infty)}{\vec x_1}' = \Xi_{[t,\infty)}{\vec x_2}'$. Similarly, using a straightforward induction argument on depth, one can show that for every neuron in the network, its output until time instant $t$ is identical in either case. We therefore have $\Xi_{[t, \infty)} {\cal T}(\chi_1) =\Xi_{[t, \infty)}{\cal T}(\chi_2)$.

\noindent
{\bf ${\cal T}(\cdot)$ is time-invariant.}  Consider arbitrary $\chi \in {\cal F}_m$  and $t\in \mathbb{R}$ with $\sigma_t(\chi) \in {\cal F}_m$. We wish to show that ${\cal T}(\sigma_t(\chi))=\sigma_t({\cal T}(\chi))$. As before, let $N_i$, for $1\leq i\leq d$, be the set of neurons in ${\cal N}$ of depth $i$, where $d$ is the depth of ${\cal N}$. Each neuron ${\mathsf N}\in N_1$ receives all its inputs from spike-trains in $\chi$. When the network receives $\chi$ and $\sigma_t(\chi)$ as input, suppose ${\mathsf N}$ receives $\chi'$ and $\sigma_t(\chi')$ respectively as input. Let ${\vec x_1}'$ and ${\vec x_2}'$ be the output produced by ${\mathsf N}$ on receiving $\chi'$ and $\sigma_t(\chi')$ as input respectively. We wish to show that ${\vec x_2}'= \sigma_t({\vec x_1}')$. Since  $\chi' \in {\cal F}_m$, there exists a $T \in \mathbb{R}^+$, so that $\Xi_{[T,\infty)}\chi' = \Xi_{[T-t,\infty)}\sigma_t(\chi') ={\vec \phi}^{m'}$, where $m'$ is the number of inputs to $\mathsf{N}$. Therefore, by Axiom (3) of the neuron, we have $\Xi_{[T,\infty)}{\vec x_1}' = \Xi_{[T-t,\infty)}{\vec x_2}' ={\vec \phi}$.  Now, for all $t' \in
 \mathbb{R}$, $\Xi_{t'} {\vec x_1}'= \langle t' \rangle$ if and only if
 $P_{\mathsf{N}}(\Xi_{(0,\Upsilon_{\mathsf{N}})}(\sigma_{t'}(\chi')), \Xi_{(0,\rho_{\mathsf{N}})}(\sigma_{t'}({\vec x_1}')) =\tau_{\mathsf{N}}$. It is therefore straightforward to make an induction argument on the spike number, starting from the oldest spike in ${\vec x_1}'$ to show that ${\vec x_1}'$ has a spike at some $t'$ iff ${\vec x_2}'$ has a spike at $t'-t$ and therefore we have ${\vec x_2}'= \sigma_t({\vec x_1}')$.  Similarly, using a straightforward induction argument on depth, one can show that for every neuron in the network, its output in the second case is a time-shifted version of the one in the first case. We therefore have ${\cal T}(\sigma_t(\chi))=\sigma_t({\cal T}(\chi))$.

\noindent
{\bf ${\cal T}(\cdot)$ is resettable.} Let $\Upsilon$ and $\rho$ be upper bounds on those parameters over all the neurons in ${\cal N}$. If $\Upsilon<\rho$, then set the value of $\Upsilon=\rho$. The claim is that for $W= d(\Upsilon+\rho)+\rho$,   ${\cal T}(\cdot)$ is $W$-resettable, where $d$ is the depth of ${\cal N}$. Consider arbitrary $\chi \in {\cal F}_m$ so that $\chi$ has a gap in the interval $(t, t+d(\Upsilon+\rho)+\rho)$, for some $t \in \mathbb{R}$. As before, let $N_i$, for $1\leq i\leq d$, be the set of neurons in ${\cal N}$ of depth $i$. Each neuron ${\mathsf N}\in N_1$ receives all its inputs from spike-trains in $\chi$. Therefore by Axiom (3) of the neuron, it is straightforward to see that the output of ${\mathsf N}$ has a gap in the interval $(t, t+(d-1)(\Upsilon+\rho)+2\rho)$. By similar arguments, we have that output of each neuron ${\mathsf N}\in N_i$, for $1\leq i\leq d$ has a gap in the interval $(t, t+(d-i)(\Upsilon+\rho)+(i+1)\rho)$. Thus, in particular, the output neuron has a gap in the interval $(t, t+(d+1)\rho)$. Since $d\geq 1$, the Gap Lemma applies, and at time instant $t$ the output of the output neuron depends on spikes in the interval $(t, t+(\Upsilon+\rho))$ of its inputs. All inputs to the output neuron have a gap in the interval $(t, t+(\Upsilon+\rho)+d\rho)$, since they have depth at most $(d-1)$. Since those inputs have a gap in the interval $(t+(\Upsilon+\rho), t+(\Upsilon+\rho)+d\rho)$, for $d\geq 2$, the Gap Lemma applies and the output neuron's output at time instant $t$ depends on outputs of the ``penultimate layer'' in the interval $(t, t+2(\Upsilon+\rho))$. Therefore by similar arguments, the output of the output neuron at time instant $t$ at most depends on inputs from $\chi$ in the interval $(t, t+d(\Upsilon+\rho))$. That is to say that ${\cal T}(\chi')$, for every $\chi'$ identical to $\chi $ in the interval $(-\infty, t+d(\Upsilon + \rho))$, has the same output as ${\cal T}(\chi)$ in the interval $[t, -\infty)$, following the corollary to the Gap Lemma. In particular, $\Xi_{(-\infty, t]} \chi$ is one such $\chi'$. We therefore have  $\Xi_{(-\infty, t]} {\cal T}(\chi) = {\cal T}(\Xi_{(-\infty, t]} \chi)$ upon noting that $\Xi_{(-\infty, t]} {\cal T}(\Xi_{(-\infty, t]} \chi) = {\cal T}(\Xi_{(-\infty, t]} \chi)$, since ${\cal T}(\cdot)$ has no spikes in $(t, \infty)$. Thus, ${\cal T}(\cdot)$ is resettable.

\end{proof}

\begin{proof}[Proof of Proposition~\ref{depth2prop1}]
Assume that the hypothesis in the proposition is true. Let ${\cal T}:{\cal F}_m\rightarrow{\cal S}$  be $W$-Resettable for some $W \in \mathbb{R}^+$.

We first show a construction for the neuron $\mathsf{O}$, prove that it obeys all the axioms of the abstract model and then show that it has the property that for every $\chi \in {\cal F}_m$, ${\cal T}(\chi)$ is consistent with ${\cal T}_{\mathsf{J}}(\chi) \sqcup \chi$ with respect to $\mathsf{O}$.

We first construct the neuron $\mathsf{O}\langle\alpha_{\mathsf{O}}, \Upsilon_{\mathsf{O}}, \rho_{\mathsf{O}}, \tau_{\mathsf{O}}, \lambda_{\mathsf{O}}, m_{\mathsf{O}}, 
P_{\mathsf{O}}:{\bar{\cal S}_{(0, \Upsilon_{\mathsf{O}})}}^{m_{\mathsf{O}}} \times \bar{\cal S}_{(0, \rho_{\mathsf{O}})} \rightarrow [\lambda_{\mathsf{O}}, \tau_{\mathsf{O}}]\rangle$. Set $\alpha_{\mathsf{O}}=\alpha$ and $\rho_{\mathsf{O}}, \tau_{\mathsf{O}} \in \mathbb{R}^+$, $\lambda_{\mathsf{O}} \in \mathbb{R}^-$  arbitrarily with $\rho_{\mathsf{O}} \geq \alpha_{\mathsf{O}}$. Set $\Upsilon_{\mathsf{O}} = \max\{U, W\}$ and $m_{\mathsf{O}}=m+1$. The function $P_{\mathsf{O}}:{\bar{\cal S}_{(0, \Upsilon_{\mathsf{O}})}}^{m_{\mathsf{O}}} \times \bar{\cal S}_{(0, \rho_{\mathsf{O}})} \rightarrow [\lambda_{\mathsf{O}}, \tau_{\mathsf{O}}]$ is constructed as follows.

For $\chi' \in {\bar{\cal S}_{(0, \Upsilon_{\mathsf{O}})}}^{m_{\mathsf{O}}}$ and $ {\vec x_0}' \in \bar{\cal S}_{(0, \rho_{\mathsf{O}})}$, set $P_{\mathsf{O}}(\chi', {\vec x_0}')=\tau_{\mathsf{O}}$ and $P_{\mathsf{O}}(\chi', {\vec \phi})=\tau_{\mathsf{O}}$ if and only if there exists $\chi \in {\cal F}_m$ and $t \in \mathbb{R}$ so that $\Xi_{t} {\cal T}(\chi) = \langle t \rangle$ and $\chi'=\Xi_{(0,\Upsilon_{\mathsf{O}})} (\sigma_{t}({\cal T}_{\mathsf{J}}(\chi)\sqcup \chi))$ and ${\vec x_0}'=\Xi_{(0,\rho_{\mathsf{O}})} (\sigma_{t}({\cal T}(\chi)))$. Everywhere else, the value of this function is set to zero.

Next, we show it obeys all of the axioms of the single neuron.

We prove that $\mathsf{O}$ satisfies Axiom (1) by showing that its contrapositive is true. Let $\chi' \in {\bar{\cal S}_{(0, \Upsilon_{\mathsf{O}})}}^{m_{\mathsf{O}}}$ and $ {\vec x_0}' \in \bar{\cal S}_{(0, \rho_{\mathsf{O}})}$  be arbitrary so that $P_{\mathsf{O}}(\chi', {\vec x_0}')=\tau_{\mathsf{O}}$. If $ {\vec x_0}'={\vec \phi}$, Axiom (1) is immediately satisfied. Thus, consider the case when $ {\vec x_0}' = \langle x_0^{1'}, x_0^{2'}, \ldots x_0^{k'} \rangle $. Then $x_0^{1'}\geq \alpha$, otherwise, from the construction of $P_{\mathsf{O}}(\cdot)$, it is immediate that there exists a $\chi \in {\cal F}_m$ with ${\cal T}(\chi) \notin {\cal S} $.

Next, we prove that $\mathsf{O}$ satisfies Axiom (2). Let $\chi' \in {\bar{\cal S}_{(0, \Upsilon_{\mathsf{O}})}}^{m_{\mathsf{O}}}$ and $ {\vec x_0}' \in \bar{\cal S}_{(0, \rho_{\mathsf{O}})}$ be arbitrary. If $P_{\mathsf{O}}(\chi', {\vec x_0}')=\tau_{\mathsf{O}}$, then it is immediate from the construction that $P_{\mathsf{O}}(\chi', {\vec \phi})=\tau_{\mathsf{O}}$. On the contrary, if $P_{\mathsf{O}}(\chi', {\vec x_0}')\neq\tau_{\mathsf{O}}$, from the construction of $\mathsf{O}$, we have $P_{\mathsf{O}}(\chi', {\vec x_0}')=0$. Then the ``tie-breaker'' condition in the hypothesis implies that $P_{\mathsf{O}}(\chi', {\vec \phi})\neq\tau_{\mathsf{O}}$. Therefore, $P_{\mathsf{O}}(\chi', {\vec \phi})=0$. Thus, Axiom (2) is satisfied either way.

With Axiom (3), we wish to show $P_{\mathsf{O}}({\vec \phi^{m+1}}, {\vec \phi})=0$. Here, we will show that $P_{\mathsf{O}}({\vec x_{\mathsf{J}}} \sqcup {\vec \phi}^m, {\vec x_0}')=0$, for all $ {\vec x_{\mathsf{J}}} \in \bar{\cal S}_{(0, \Upsilon_{\mathsf{O}})}$ and ${\vec x_0}' \in \bar{\cal S}_{(0, \rho_{\mathsf{O}})}$ which implies the required result. Assume, for the sake of contradiction, that there exists a $ {\vec x_{\mathsf{J}}} \in \bar{\cal S}_{(0, \Upsilon_{\mathsf{O}})}$ and ${\vec x_0}' \in \bar{\cal S}_{(0, \rho_{\mathsf{O}})}$, so that  $P_{\mathsf{O}}({\vec x_{\mathsf{J}}} \sqcup {\vec \phi}^m, {\vec x_0}')=\tau_{\mathsf{O}}$. From the construction of $\mathsf{O}$, this implies that there exists $\chi \in {\cal F}_m$ and $t \in \mathbb{R}$ so that $\Xi_{t} {\cal T}(\chi) = \langle t \rangle$ and $\Xi_{(0,\Upsilon_{\mathsf{O}})} (\sigma_{t}(\chi))= {\vec \phi}^m$. That is, $\chi$ has a gap in the interval $(t, t+W)$, since $\Upsilon_{\mathsf{O}} \geq W$. Since ${\cal T}:{\cal F}_m\rightarrow{\cal S}$ is causal, time-invariant and $W$-resettable, by Corollary \ref{corr3} (stated and proved later in the present write-up), we have $\Xi_{t} {\cal T}(\chi) = {\vec \phi}$ , which is a contradiction. Therefore, we have $P_{\mathsf{O}}({\vec x_{\mathsf{J}}} \sqcup {\vec \phi}^m, {\vec x_0}')\neq\tau_{\mathsf{O}}$ and by construction of $\mathsf{O}$, $P_{\mathsf{O}}({\vec x_{\mathsf{J}}} \sqcup {\vec \phi}^m, {\vec x_0}')=0$, for all $ {\vec x_{\mathsf{J}}} \in \bar{\cal S}_{[0, \Upsilon_{\mathsf{O}}]}$ and ${\vec x_0}' \in \bar{\cal S}_{[0, \rho_{\mathsf{O}}]}$. This implies $P_{\mathsf{O}}({\vec \phi^{m+1}}, {\vec \phi})=0$, satisfying Axiom (3).

Finally, we wish to show that for every $\chi \in {\cal F}_m$, ${\cal T}(\chi)$ is consistent with ${\cal T}_{\mathsf{J}}(\chi) \sqcup \chi$ with respect to $\mathsf{O}$. That is, we wish to show that for every $\chi \in {\cal F}_m$ and for every $t \in \mathbb{R}$, $\Xi_{0} \sigma_t({\cal T}(\chi))= \langle 0 \rangle$ if and only if $P_{\mathsf{O}}(\Xi_{(0,\Upsilon_{\mathsf{O}})} (\sigma_{t}({\cal T}_{\mathsf{J}}(\chi)\sqcup \chi)), \Xi_{(0,\rho_{\mathsf{O}})} (\sigma_{t}({\cal T}(\chi))))=\tau_{\mathsf{O}}$. Consider arbitrary $\chi \in {\cal F}_m$ and $t \in \mathbb{R}$. If $\Xi_{0} \sigma_t({\cal T}(\chi))= \langle 0 \rangle$, then it is immediate from the construction of $\mathsf{O}$ that $P_{\mathsf{O}}(\Xi_{(0,\Upsilon_{\mathsf{O}})} (\sigma_{t}({\cal T}_{\mathsf{J}}(\chi)\sqcup \chi)), \Xi_{(0,\rho_{\mathsf{O}})} (\sigma_{t}({\cal T}(\chi))))=\tau_{\mathsf{O}}$. To prove the converse, suppose $\Xi_{0} \sigma_t({\cal T}(\chi))\neq \langle 0 \rangle$. Then, from the contrapositive of the ``tie-breaker'' condition, it follows that for all $\tilde{\chi} \in {\cal F}_m$ and for all $\tilde{t} \in \mathbb{R}$ with $\Xi_{(0,\Upsilon_{\mathsf{O}})} (\sigma_{\tilde{t}}({\cal T}_{\mathsf{J}}(\tilde{\chi})\sqcup \tilde{\chi}))=\Xi_{(0,\Upsilon_{\mathsf{O}})} (\sigma_{t}({\cal T}_{\mathsf{J}}(\chi)\sqcup \chi))$, we have $\Xi_0\sigma_{\tilde{t}}({\cal T}(\tilde{\chi})) = \Xi_0\sigma_{t}({\cal T}(\chi))\neq \langle 0 \rangle$. Therefore, from the construction, we have $P_{\mathsf{O}}(\Xi_{(0,\Upsilon_{\mathsf{O}})} (\sigma_{t}({\cal T}_{\mathsf{J}}(\chi)\sqcup \chi)), \Xi_{(0,\rho_{\mathsf{O}})} (\sigma_{t}({\cal T}(\chi))))\neq\tau_{\mathsf{O}}$.

\end{proof}

\begin{proof}[Proof of Proposition~\ref{depth2prop2}]
Assume that the hypothesis in the proposition is true. Let ${\cal T}:{\cal F}_m\rightarrow{\cal S}$  be $W'$-Resettable for some $W' \in \mathbb{R}^+$. Set $W=\max\{W', 12\alpha\}$. One readily verifies that ${\cal T}:{\cal F}_m\rightarrow{\cal S}$ is also $W$-resettable.

We first show a construction for the neuron $\mathsf{J}$, prove that it obeys all the axioms and then show that it has the property that there exists a $U \in \mathbb{R}^+$ so that for all $t_1, t_2 \in \mathbb{R}$ and $\chi_1, \chi_2 \in {\cal F}_m$ with $ \Xi_{0} \sigma_{t_1}({\cal T}(\chi_1))\neq \Xi_{0} \sigma_{t_2}({\cal T}(\chi_2))$, we have $\Xi_{(0,U)} (\sigma_{t_1}({\cal T}_I(\chi_1)\sqcup \chi_1))\neq \Xi_{(0,U)} (\sigma_{t_2}({\cal T}_{\mathsf{J}}(\chi_2)\sqcup \chi_2))$, where ${\cal T}_{\mathsf{J}}:{\cal F}_m\rightarrow{\cal S}$ is such that for each $\chi \in {\cal F}_m$, ${\cal T}_{\mathsf{J}}(\chi)$ is consistent with $\chi$ with respect to $\mathsf{J}$. 

We first construct the neuron $\mathsf{J}\langle\alpha_{\mathsf{J}}, \Upsilon_{\mathsf{J}}, \rho_{\mathsf{J}}, \tau_{\mathsf{J}}, \lambda_{\mathsf{J}}, m_{\mathsf{J}}, P_{\mathsf{J}}:{\bar{\cal S}_{(0, \Upsilon_{\mathsf{J}})}}^{m_{\mathsf{J}}} \times \bar{\cal S}_{(0, \rho_{\mathsf{J}})} \rightarrow [\lambda_{\mathsf{J}}, \tau_{\mathsf{J}}]\rangle$. Set $\alpha_{\mathsf{J}}=\alpha$. Let $p, q, r \in \mathbb{R}^+$, with\footnote{The choice of values for $p$, $q$, $r$ and $W$ was made so as to satisfy the following inequalities, which we will need in the proof: $p<W, p>2(q+r)$ and $q>\alpha$.} $p=8\alpha, q=2\alpha$ and $r=\alpha$. Set $\Upsilon_{\mathsf{J}}= p+q+r+W$, $\rho_{\mathsf{J}}=2p-r$ and $m_{\mathsf{J}}=m$. Let $\tau_{\mathsf{J}} \in \mathbb{R}^+$, $\lambda_{\mathsf{J}} \in \mathbb{R}^-$  be chosen arbitrarily. The function $P_{\mathsf{J}}:{\bar{\cal S}_{(0, \Upsilon_{\mathsf{J}})}}^{m_{\mathsf{J}}} \times \bar{\cal S}_{(0, \rho_{\mathsf{J}})} \rightarrow [\lambda_{\mathsf{J}}, \tau_{\mathsf{J}}]$ is constructed as follows.

For $\chi \in {\bar{\cal S}_{(0, \Upsilon_{\mathsf{J}})}}^{m_{\mathsf{J}}}$ and $ {\vec x_0} \in \bar{\cal S}_{(0, \rho_{\mathsf{J}})}$, set $P_{\mathsf{J}}(\chi, {\vec x_0})=\tau_{\mathsf{J}}$ and $P_{\mathsf{J}}(\chi, {\vec \phi})=\tau_{\mathsf{J}}$ 
if and only if one of the following is true; everywhere else, the function is set to zero.
\begin{enumerate}
\item $\Xi_{(p, p+W)} \chi = {\vec \phi}^{m_{\mathsf{J}}}$, $\Xi_{p} \chi \neq {\vec \phi}^{m_{\mathsf{J}}}$  and $\Xi_{(0, p]} {\vec x_0} = {\vec \phi}$.

\item $\Xi_{(0, p+q]} {\vec x_0} = \langle t \rangle$, where $q\leq t <(q+r)$ and $(t-q)=\varepsilon(0, \Xi_{(0,p]}\sigma_t(\chi))$. Moreover, $\Xi_{(t+p, t+p+W)} \chi = {\vec \phi}^{m_{\mathsf{J}}}$ and $\Xi_{(p+t)} \chi \neq {\vec \phi}^{m_{\mathsf{J}}}$.

\item $\Xi_{(0, 2p-(q+r)]} {\vec x_0} = \langle t_x, t_y \rangle$ with $ (p-(q+r)) < t_x \leq (p-q) \leq t_y = p$. Also, for all $t'\in [0,p]$, $\Xi_{(t', t'+W)} \chi \neq {\vec \phi}^{m_{\mathsf{J}}}$.

\item $\Xi_{[0, 2p-r]} {\vec x_0} = \langle t, t_x, t_y \rangle$ with $ q\leq t < (q+r) < (p-r) \leq t_x \leq p \leq t_y = p+t$ and $(t-q)=\varepsilon((t_y-t_x-q), \Xi_{(0,p]}\sigma_t(\chi))$.  Furthermore, for all $t'\in [0,p+t]$, $\Xi_{(t', t'+W]} \chi \neq {\vec \phi}^{m_{\mathsf{J}}}$.

\end{enumerate}
\noindent
where $\varepsilon:[0,r)\times{\bar{\cal S}_{(0, p]}}^{m_{\mathsf{J}}} \rightarrow [0,r)$ is as defined below.

For convenience, we define an operator $\iota_j^k:[0,1)\rightarrow [0,1)$, for $j,k \in \mathbb{Z}^+$, that constructs a new number obtained by concatenating every $i^{th}$ digit of a given number, where $i\equiv j \bmod k$. More formally, for $x \in [0,1)$, $\iota_j^k(x)=\Sigma_{i=1}^{\infty}  ((\lfloor  x \times 10^{j+(i-1)k}\rfloor-10\lfloor x \times 10^{j+(i-1)k-1}  \rfloor) \times 10^{-i})$.

Also, we define another operator $\zeta_k:[0,1)^k\rightarrow [0,1)$, for $k \in \mathbb{Z}^+$ which ``interleaves'' the digits of $k$ given numbers in order to produce a new number. More formally, for $x_0, x_1, \ldots, x_{k-1} \in [0,1)$, $\zeta_k(x_0, x_1, \ldots, x_{k-1}) = \Sigma_{i=0}^{\infty} ((\lfloor  x_{k(\frac{i}{k}-\lfloor \frac{i}{k} \rfloor)} \times 10^{1+\lfloor\frac{i}{k}\rfloor}\rfloor-10\lfloor  x_{k(\frac{i}{k}-\lfloor \frac{i}{k} \rfloor)} \times 10^{\lfloor\frac{i}{k}\rfloor}\rfloor) \times 10^{-(i+1)})$.

Let $d$ be the largest integer so that, for all $x' \in [0,r)$, we have $x' \times 10^d <1$. For $x' \in [0,r)$, let $x= x' \times 10^{d}$. For $\chi \in {\bar{\cal S}_{(0, p]}}^{m_{\mathsf{J}}}$, define\footnote{Recall that the {\em projection operator for spike-train ensembles} is defined as $\Pi_i(\chi) = {\vec x_i}$, for
    $1\leq i \leq m$, where $\chi=\langle {\vec x_1}, {\vec x_2}, \ldots, {\vec x_m} \rangle$.} $\varepsilon(x', \chi)= 10^{-d} \times \zeta_{m_{\mathsf{J}}}(\varepsilon_0(\iota_1^{m_{\mathsf{J}}}(x),\Pi_1(\chi)),  \varepsilon_0(\iota_2^{m_{\mathsf{J}}}(x),\Pi_2(\chi)), \ldots, \varepsilon_0(\iota_{m_{\mathsf{J}}}^{m_{\mathsf{J}}}(x),\Pi_{m_{\mathsf{J}}}(\chi)))$, where $\varepsilon_0:[0,1)\times {\bar{\cal S}_{(0, p]}} \rightarrow [0,1)$ is as defined below.

Let $n \in [0,1)$ and ${\vec x} \in {\bar{\cal S}_{(0, p]}}$. Furthermore, let $c=\iota_1^2(n)$ and $s=\iota_2^2(n)$. Let ${\vec x}=\langle x^1, x^2,\ldots, x^k\rangle$. We have $0\leq k \leq 8$, because $p=8\alpha$. Also, since $p=8r$, we have $x^i \times 10^{d-1} < 1$, for $1\leq i \leq k$. Let $s'=\zeta_{k+1}(x^1\times 10^{d-1}, x^2\times 10^{d-1}, \ldots, x^k\times 10^{d-1}, s)$. If $c=0$, then let $c'=\frac{k}{10}+0.09$ else let $c'=\frac{k}{10}+\frac{c}{10}$. Finally, define $\varepsilon_0(n, {\vec x})=\zeta_2(c',s')$.

Next, we show that $\mathsf{J}$ satisfies all the axioms of the neuron.

It is immediate that $\mathsf{J}$ satisfies Axiom (1), since all output spikes in the above construction are at least $q$ milliseconds apart, and $q=2\alpha$.

We now prove that $\mathsf{J}$ satisfies Axiom (2). Let $\chi' \in {\bar{\cal S}_{(0, \Upsilon_{\mathsf{J}})}}^{m_{\mathsf{J}}}$ and $ {\vec x_0}' \in \bar{\cal S}_{(0, \rho_{\mathsf{J}})}$ be arbitrary. If $P_{\mathsf{J}}(\chi', {\vec x_0}')=\tau_{\mathsf{J}}$, then it is immediate from the construction that $P_{\mathsf{J}}(\chi', {\vec \phi})=\tau_{\mathsf{J}}$ which satisfies Axiom (2). On the contrary, if $P_{\mathsf{J}}(\chi', {\vec x_0}')\neq\tau_{\mathsf{J}}$, from the construction of $\mathsf{J}$, we have $P_{\mathsf{J}}(\chi', {\vec x_0}')=0$. Also, from the construction we have either $P_{\mathsf{J}}(\chi', {\vec \phi})=0$ or $P_{\mathsf{J}}(\chi', {\vec \phi})=\tau_{\mathsf{J}}$. Axiom (2) is satisfied in either case.

Also, $\mathsf{J}$ satisfies Axiom (3), since it is clear that $\chi = {\vec \phi}^{m_{\mathsf{J}}}$ does not satisfy any of the conditions enumerated above. We therefore have $P_{\mathsf{J}}({\vec \phi}^{m_{\mathsf{J}}}, {\vec \phi})=0$.

Finally, we show that there exists a $U \in \mathbb{R}^+$ so that for all $t_1, t_2 \in \mathbb{R}$ and $\chi_1, \chi_2 \in {\cal F}_m$ with $ \Xi_{0} \sigma_{t_1}({\cal T}(\chi_1))\neq \Xi_{0} \sigma_{t_2}({\cal T}(\chi_2))$, we have $\Xi_{(0,U)} (\sigma_{t_1}({\cal T}_{\mathsf{J}}(\chi_1)\sqcup \chi_1))\neq \Xi_{(0,U)} (\sigma_{t_2}({\cal T}_{\mathsf{J}}(\chi_2)\sqcup \chi_2))$, where ${\cal T}_{\mathsf{J}}:{\cal F}_m\rightarrow{\cal S}$ such that for each $\chi \in {\cal F}_m$, ${\cal T}_{\mathsf{J}}(\chi)$ is consistent with $\chi$ with respect to $\mathsf{J}$. Let $U=p+q+r+W$. Assume $ \Xi_{0} \sigma_{t_1}({\cal T}(\chi_1))\neq \Xi_{0} \sigma_{t_2}({\cal T}(\chi_2))$. Now, suppose $\Xi_{(0, 0+W)} \sigma_{t_1}(\chi_1) = {\vec \phi}^m$, then clearly $\Xi_{(0, 0+W)} \sigma_{t_2}(\chi_2) \neq {\vec \phi}^m$, otherwise ${\cal T}(\cdot)$ produces no spike at times $t_1$ and $t_2$ respectively on receiving $\chi_1$ and $\chi_2$, by Corollary~\ref{corr3}. As a result, $\Xi_{(0, U)} \sigma_{t_1}(\chi_1) \neq \Xi_{(0, U)} \sigma_{t_2}(\chi_2)$, which implies the required result. Otherwise, from Proposition \ref{prop8}, it follows that there exist $V_1, V_2 \in \mathbb{R}^+$ so that $\Xi_{(0, V_1]} (\sigma_{t_1}(\chi_1)) \neq \Xi_{(0, V_2]} (\sigma_{t_2}(\chi_2))$. If $\Xi_{(0, U)} (\sigma_{t_1}(\chi_1)) \neq \Xi_{(0, U)} (\sigma_{t_2}(\chi_2))$, it is immediate that $\Xi_{(0,U)} (\sigma_{t_1}({\cal T}_I(\chi_1)\sqcup \chi_1))\neq \Xi_{(0,U)} (\sigma_{t_2}({\cal T}_{\mathsf{J}}(\chi_2)\sqcup \chi_2))$. It therefore suffices to prove that if  $\Xi_{[U, V_1]} (\sigma_{t_1}(\chi_1)) \neq \Xi_{[U, V_2]} (\sigma_{t_2}(\chi_2))$ then $\Xi_{(0,U)} (\sigma_{t_1}{\cal T}_{\mathsf{J}}(\chi_1))\neq \Xi_{(0,U)} (\sigma_{t_2}{\cal T}_{\mathsf{J}}(\chi_2))$. Proposition \ref{prop8} implies that $\Xi_{(V_1, V_1 +W)} (\sigma_{t_1}(\chi_1)) = {\vec \phi}^m$ and $\Xi_{V_1} (\sigma_{t_1}(\chi_1))\neq {\vec \phi}^m$. Therefore, by Case (1) of the construction, $\Xi_{(V_1-p)} \sigma_{t_1}{\cal T}_{\mathsf{J}}(\chi_1) = \langle V_1-p \rangle$. Moreover, since Proposition \ref{prop8} implies that for all $t_1' \in [0, V_1)$, $\Xi_{(t_1', t_1' +W)} (\sigma_{t_1}(\chi_1)) \neq {\vec \phi}^m$,  from Case (3) of the construction,  we have that for every $k \in \mathbb{Z}^+$ with $V_1-kp >0$, $\Xi_{(V_1-kp)} \sigma_{t_1}{\cal T}_{\mathsf{J}}(\chi_1) = \langle V_1-kp \rangle$. Let $k_1$ be\footnote{$k_1$ exists because $U>p$.} the smallest positive integer, so that $V_1-k_1p < U$. From the previous arguments, we have $\Xi_{(V_1-k_1p)} \sigma_{t_1}{\cal T}_{\mathsf{J}}(\chi_1) = \langle V_1-k_1p \rangle$. Also, it is easy to see that $V_1-k_1p \geq (q+r)$. Let $k_2$ be similarly defined with respect to $\chi_2$ so that $\Xi_{(V_2-k_2p)} \sigma_{t_2}{\cal T}_{\mathsf{J}}(\chi_2) = \langle V_2-k_2p \rangle$ and $V_2-k_2p < U$. Now, there are two cases:
\begin{enumerate}
\item If $V_1-k_1p \neq V_2-k_2p$, we  now show that $\Xi_{(0,U)} (\sigma_{t_1}{\cal T}_{\mathsf{J}}(\chi_1))\neq \Xi_{(0,U)} (\sigma_{t_2}{\cal T}_{\mathsf{J}}(\chi_2))$, which is the required result. Assume, without loss of generality, that $V_1-k_1p < V_2-k_2p$. If these two quantities are less than $p-r$ apart, we have $\Xi_{(0,U)} (\sigma_{t_1}{\cal T}_{\mathsf{J}}(\chi_1))\neq \Xi_{(0,U)} (\sigma_{t_2}{\cal T}_{\mathsf{J}}(\chi_2))$, because by Case (4) of the construction ${\cal T}_{\mathsf{J}}(\chi_1)$ has a spike in the interval $(V_1-k_1p -(q+r), V_1-k_1p -q]$ and  by Case (3) of the construction, ${\cal T}_{\mathsf{J}}(\chi_2)$ has no spike in the interval $(V_2-k_2p, V_2-k_2p+p-(q+r))$. In other words, the spike following the one at $V_1-k_1p$ in ${\cal T}_{\mathsf{J}}(\chi_1)$ has no counterpart in ${\cal T}_{\mathsf{J}}(\chi_2)$. On the other hand, if they are less than $p$ apart but at most $p-r$ apart, by similar arguments, it is easy to show that the spike at $V_2-k_2p$ in  ${\cal T}_{\mathsf{J}}(\chi_2)$ has no counterpart in ${\cal T}_{\mathsf{J}}(\chi_1)$. Finally, if they are at least $p$ apart, then $k_2$ does not satisfy the property that it is the smallest positive integer, so that $V_2-k_2p \leq U$, which is a contradiction.

\item On the contrary, consider the case when $V_1-k_1p = V_2-k_2p$. We have two cases:
\begin{enumerate}
\item Suppose $k_1\neq k_2$. Let $t_1'$ be the largest positive integer so that $\Xi_{t_1'} \sigma_{t_1}{\cal T}_{\mathsf{J}}(\chi_1) = \langle t_1' \rangle$ and $t_1'< V_1-k_1p$. From Case (4) of the construction, we have that $q\leq (V_1-k_1p) - t_1'\leq q+r$. Let $t_2'$ be defined likewise, with respect to $\chi_2$. Further, let $n_1'= (V_1-k_1p) - t_1' - q$ and $n_2'= (V_2-k_2p) - t_2' - q$ and $n_1=n_1' \times 10^d$ and $n_2=n_2' \times 10^d$. Since $k_1 \neq k_2$, it is straightforward to verify that for all $j$ with $1\leq j \leq m_{\mathsf{J}}$, $\iota_1^2(\iota_j^{m_{\mathsf{J}}}(n_1))\neq \iota_1^2(\iota_j^{m_{\mathsf{J}}}(n_2))$, for the former number has $9$ in the $(k_1+1)^{th}$ decimal place, while the latter number does in the $(k_2+1)^{th}$ decimal place and not in the $(k_1+1)^{th}$ decimal place since $k_1\neq k_2$. Therefore, $n_1 \neq n_2$ and consequently $t_1'\neq t_2'$ which gives us $\Xi_{(0,U)} (\sigma_{t_1}{\cal T}_{\mathsf{J}}(\chi_1))\neq \Xi_{(0,U)} (\sigma_{t_2}{\cal T}_{\mathsf{J}}(\chi_2))$, which is the required result. 

\item On the other hand, suppose $k_1=k_2$. Again, we have two cases:
\begin{enumerate}
\item Suppose, there exists a $j$ with $1\leq j \leq m_{\mathsf{J}}$ and a $k'\leq k_1$, so that $\Xi_{(V_1-k'p, V_1-(k'-1)p]}\Pi_j(\sigma_{t_1}(\chi_1))$ has a different number of spikes when compared to $\Xi_{(V_2-k'p, V_2-(k'-1)p]}\Pi_j(\sigma_{t_2}(\chi_2))$. Let $n_1, n_2$ be defined as before. It is straightforward to verify that $\iota_1^2(\iota_j^{m_{\mathsf{J}}}(n_1))\neq \iota_1^2(\iota_j^{m_{\mathsf{J}}}(n_2))$, because they differ in the $(k_1 - k'+1)^{th}$ decimal place\footnote{Which in $n_1$ and $n_2$ encodes the number of spikes in the interval $(V_2-k'p, V_2-(k'-1)p]$ on the $j^{th}$ spike-train of $\chi_1$ and $\chi_2$ respectively.} . Therefore, $\Xi_{(0,U)} (\sigma_{t_1}{\cal T}_{\mathsf{J}}(\chi_1))\neq \Xi_{(0,U)} (\sigma_{t_2}{\cal T}_{\mathsf{J}}(\chi_2))$.

\item Now consider the case where for all $j$ with $1\leq j \leq m_{\mathsf{J}}$ and $k'\leq k_1$, we have $\Xi_{(V_1-k'p, V_1-(k'-1)p]}\Pi_j(\sigma_{t_1}(\chi_1))$ have the same number of spikes when compared to $\Xi_{(V_2-k'p, V_2-(k'-1)p]}\Pi_j(\sigma_{t_2}(\chi_2))$. Now, by hypothesis, we have $\Xi_{[U, V_1]} (\sigma_{t_1}(\chi_1)) \neq \Xi_{[U, V_2]} (\sigma_{t_2}(\chi_2))$. Therefore there must exist a $1\leq j \leq m_{\mathsf{J}}$ and $k'\leq k_1$, so that there is a point in time where one of the spike-trains $\Xi_{(V_1-k'p, V_1-(k'-1)p]}\Pi_j(\sigma_{t_1}(\chi_1))$ and $\Xi_{(V_2-k'p, V_2-(k'-1)p]}\Pi_j(\sigma_{t_2}(\chi_2))$ has a spike, while the other does not. Let $t'$ be the latest time instant at which this is so. Also, assume without loss of generality that $\Xi_{(V_1-k'p, V_1-(k'-1)p]}\Pi_j(\sigma_{t_1}(\chi_1)) = \langle x^1, \ldots, x^q \rangle$ has a spike at time instant $t'$ while $\Xi_{(V_2-k'p, V_2-(k'-1)p]}\Pi_j(\sigma_{t_2}(\chi_2))$ does not. Let $p$ be the number so that $t'=x^p$. Let $n_1, n_2$ be defined as before.  Also, for each $h$ with $1\leq h \leq k_1$, let $r_h$ be the number of spikes in $\Xi_{(V_1-hp, V_1-(h-1)p]}\Pi_j(\sigma_{t_1}(\chi_1))$. Each $r_h$ can be determined from $n_1$. Then, it is straightforward to verify\footnote{The expression on either side of the inequality is a real number that encodes for the $p^{th}$ spike time in the spike-trains $\Xi_{(V_1-k'p, V_1-(k'-1)p]}\Pi_j(\sigma_{t_1}(\chi_1))$ and $\Xi_{(V_2-k'p, V_2-(k'-1)p]}\Pi_j(\sigma_{t_2}(\chi_2))$ respectively.} that $\iota_{p}^{r_{k'}} \iota_{r_{k'-1}}^{r_{k'-1}}\ldots \iota_{r_1}^{r_1} \iota_2^2  \iota_j^{m_{\mathsf{J}}} n_1 \neq \iota_{p}^{r_{k'}} \iota_{r_{k'-1}}^{r_{k'-1}}\ldots \iota_{r_1}^{r_1} \iota_2^2  \iota_j^{m_{\mathsf{J}}} n_2$. Therefore, $n_1\neq n_2$ and it follows that \\$\Xi_{(0,U)} (\sigma_{t_1}{\cal T}_{\mathsf{J}}(\chi_1))\neq \Xi_{(0,U)} (\sigma_{t_2}{\cal T}_{\mathsf{J}}(\chi_2))$.

\end{enumerate}
\end{enumerate}
\end{enumerate}

\end{proof}

\subsubsection*{Some auxiliary propositions used in the proofs of Propositions \ref{depth2prop1} and \ref{depth2prop2}}
\addtocounter{proposition}{2}
\begin{proposition}
If ~${\cal T}:{\cal F}_m\rightarrow{\cal S}$ is time-invariant, then ${\cal T}({\vec \phi}^m)={\vec \phi}$.
\end{proposition}
\begin{proof}
For the sake of contradiction, suppose ${\cal T}({\vec \phi}^m)={\vec x_0}$, where ${\vec x_0} \neq {\vec \phi}$. That is, there exists a $t \in \mathbb{R}$ with $\Xi_t {\vec x_0} = \langle t \rangle$. Let $\delta < \alpha$. Clearly, $\sigma_\delta({\vec \phi}^m) = {\vec \phi}^m \in {\cal F}_m$. Since ${\cal T}:{\cal F}_m\rightarrow{\cal S}$  is time-invariant,  ${\cal T}(\sigma_\delta({\vec \phi}^m)) = \sigma_\delta({\cal T}({\vec \phi}^m)) = \sigma_\delta({\vec x_0})$. Now, $\sigma_\delta({\vec x_0}) \neq {\vec x_0}$ since $\Xi_{(t-\delta)} \sigma_{\delta}({\vec x_0}) = \langle t-\delta \rangle$ whereas $\Xi_{(t-\delta)}{\vec x_0} = {\vec \phi}$, for otherwise ${\vec x_0} \notin {\cal S}$. This is a contradiction. Therefore, ${\cal T}({\vec \phi}^m)={\vec \phi}$.

\end{proof}

\addtocounter{corollary}{2}
\begin{corollary}\label{corr3}
Let ~${\cal T}:{\cal F}_m\rightarrow{\cal S}$ be causal, time-invariant and $W$-resettable, for some $W \in \mathbb{R}^+$. If $\chi \in {\cal F}_m$ has a gap in the interval $(t, t+W)$, then $\Xi_t {\cal T}(\chi) ={\vec \phi}$.
\end{corollary}
\begin{proof}
Assume the hypothesis of the above statement. One readily sees that $\Xi_t {\cal T}(\chi) = \Xi_{[t, \infty)} \Xi_{(-\infty, t]} {\cal T}(\chi)$. Now, since $\chi$ has a gap in the interval $(t, t+W)$ and ${\cal T}:{\cal F}_m\rightarrow{\cal S}$  is $W$-resettable, we have $\Xi_{[t, \infty)} \Xi_{(-\infty, t]} {\cal T}(\chi) = \Xi_{[t, \infty)} {\cal T}(\Xi_{(-\infty, t]} \chi)$. Further, by definition, $\Xi_{(t,\infty)} \Xi_{(-\infty, t]} \chi = \Xi_{(t,\infty)} {\vec \phi}^m$. Therefore, since ${\cal T}:{\cal F}_m\rightarrow{\cal S}$ is causal, it follows that $\Xi_{[t, \infty)} {\cal T}(\Xi_{(-\infty, t]} \chi) = \Xi_{[t, \infty)} {\cal T}({\vec \phi}^m)={\vec \phi}$, with the last equality following from the previous proposition. Thus, we have $\Xi_t {\cal T}(\chi) ={\vec \phi}$.
\end{proof}

\begin{proposition}\label{prop8}
Let ~${\cal T}:{\cal F}_m\rightarrow{\cal S}$ be causal, time-invariant and $W'$-resettable, for some $W' \in \mathbb{R}^+$. Then for all $W \in \mathbb{R}^+$ with $W \geq W'$, $t_1, t_2 \in \mathbb{R}$ and $\chi_1, \chi_2 \in {\cal F}_m$ with $ \Xi_{0} \sigma_{t_1}({\cal T}(\chi_1))\neq \Xi_{0} \sigma_{t_2}({\cal T}(\chi_2))$, where $\Xi_{(0, 0+W)} \sigma_{t_1}(\chi_1) \neq {\vec \phi}^m \neq \Xi_{(0, 0+W)} \sigma_{t_2}(\chi_2)$, there exist $V_1, V_2 \in \mathbb{R}^+$ so that the following are true.
\begin{enumerate}
\item $\Xi_{(0, V_1]} (\sigma_{t_1}(\chi_1)) \neq \Xi_{(0, V_2]} (\sigma_{t_2}(\chi_2))$ 
\item $\Xi_{(V_1, V_1 +W)} (\sigma_{t_1}(\chi_1)) = {\vec \phi}^m$, $\Xi_{V_1} (\sigma_{t_1}(\chi_1))\neq {\vec \phi}^m$ and $\Xi_{(V_2, V_2+W)} (\sigma_{t_2}(\chi_2)) = {\vec \phi}^m$, $\Xi_{V_2} (\sigma_{t_2}(\chi_2)) \neq {\vec \phi}^m$
\item For all $t_1' \in [0, V_1)$, $\Xi_{(t_1', t_1' +W)} (\sigma_{t_1}(\chi_1)) \neq {\vec \phi}^m$ and for all $t_2' \in [0, V_2)$,\\ $\Xi_{(t_2', t_2' +W)} (\sigma_{t_2}(\chi_2)) \neq {\vec \phi}^m$.

\end{enumerate}
\end{proposition}

\begin{proof}
Since ${\cal T}:{\cal F}_m\rightarrow{\cal S}$ is causal, we have $\Xi_{[t_1, \infty)} {\cal T}(\chi_1)  = \Xi_{[t_1, \infty)} {\cal T}(\Xi_{(t_1, \infty)} \chi_1)$. This implies $\sigma_{t_1}(\Xi_{[t_1, \infty)} {\cal T}(\chi_1))  = \sigma_{t_1}(\Xi_{[t_1, \infty)} {\cal T}(\Xi_{(t_1, \infty)} \chi_1))$ which gives us \\$\Xi_{[0, \infty)}\sigma_{t_1}({\cal T}(\chi_1))  = \Xi_{[0, \infty)}\sigma_{t_1}({\cal T}(\Xi_{(t_1, \infty)} \chi_1))$. Since ${\cal T}:{\cal F}_m\rightarrow{\cal S}$ is \\time-invariant and $\sigma_{t_1}(\Xi_{(t_1, \infty)} \chi_1)=\Xi_{(0, \infty)} \sigma_{t_1}(\chi_1) \in {\cal F}_m$, we have \\$\Xi_{[0, \infty)}\sigma_{t_1}({\cal T}(\Xi_{(t_1, \infty)} \chi_1))= \Xi_{[0, \infty)} {\cal T}(\Xi_{(0, \infty)} \sigma_{t_1}(\chi_1))$. In short, \\$\Xi_{[0, \infty)}\sigma_{t_1}({\cal T}(\chi_1))  = \Xi_{[0, \infty)} {\cal T}(\Xi_{(0, \infty)} \sigma_{t_1}(\chi_1))$ which implies \\$\Xi_{0}\sigma_{t_1}({\cal T}(\chi_1))  = \Xi_{0} {\cal T}(\Xi_{(0, \infty)} \sigma_{t_1}(\chi_1))$. Similarly, $\Xi_{0}\sigma_{t_2}({\cal T}(\chi_2))  = \Xi_{0} {\cal T}(\Xi_{(0, \infty)} \sigma_{t_2}(\chi_2))$. Therefore, it follows from the hypothesis that $ \Xi_{0} {\cal T}(\Xi_{(0, \infty)} (\sigma_{t_1}(\chi_1)))\neq \Xi_{0} {\cal T}(\Xi_{(0, \infty)} (\sigma_{t_2}(\chi_2)))$.

Let $V_1, V_2 \in \mathbb{R}^+$ be the smallest positive real numbers so that $\Xi_{(0, \infty)} (\sigma_{t_1}(\chi_1))$ and $\Xi_{(0, \infty)} (\sigma_{t_2}(\chi_2))$ have gaps in the intervals $(V_1, V_1+W)$ and $(V_2, V_2+W)$ respectively. That such $V_1, V_2$ exist follows from the fact that $\chi_1, \chi_2 \in {\cal F}_m$. Since, ${\cal T}:{\cal F}_m\rightarrow{\cal S}$  is $W'$-resettable, it is also $W$-resettable for $W\geq W'$. It therefore follows that $\Xi_{(-\infty, V_1]} {\cal T}(\Xi_{(0, \infty)} (\sigma_{t_1}(\chi_1))) = {\cal T}(\Xi_{(-\infty, V_1]} \Xi_{(0, \infty)} (\sigma_{t_1}(\chi_1)))$ which equals ${\cal T}(\Xi_{(0, V_1]} (\sigma_{t_1}(\chi_1)))$. This implies that $\Xi_0 \Xi_{(-\infty, V_1]} {\cal T}(\Xi_{(0, \infty)} (\sigma_{t_1}(\chi_1)))= \Xi_0 {\cal T}(\Xi_{(0, V_1]} (\sigma_{t_1}(\chi_1)))$ due to which we have $\Xi_0 {\cal T}(\Xi_{(0, \infty)} (\sigma_{t_1}(\chi_1)))= \Xi_0 {\cal T}(\Xi_{(0, V_1]} (\sigma_{t_1}(\chi_1)))$. Likewise, \\$\Xi_0 {\cal T}(\Xi_{(0, \infty)} (\sigma_{t_2}(\chi_2)))= \Xi_0 {\cal T}(\Xi_{(0, V_2]} (\sigma_{t_2}(\chi_2)))$. We therefore have $\Xi_0 {\cal T}(\Xi_{(0, V_1]} (\sigma_{t_1}(\chi_1))) \neq \Xi_0 {\cal T}(\Xi_{(0, V_2]} (\sigma_{t_2}(\chi_2)))$. This readily implies $\Xi_{(0, V_1]} (\sigma_{t_1}(\chi_1)) \neq \Xi_{(0, V_2]} (\sigma_{t_2}(\chi_2))$ and, from the construction, it follows that $\Xi_{(V_1, V_1 +W)} (\sigma_{t_1}(\chi_1)) = {\vec \phi}^m$, $\Xi_{V_1} (\sigma_{t_1}(\chi_1))\neq {\vec \phi}^m$ and $\Xi_{(V_2, V_2+W)} (\sigma_{t_2}(\chi_2)) = {\vec \phi}^m$, $\Xi_{V_2} (\sigma_{t_2}(\chi_2)) \neq {\vec \phi}^m$, for otherwise $V_1$ or $V_2$ would not be the smallest choice of numbers with the said property. Furthermore, for the same reasons, for all $t_1' \in [0, V_1)$, $\Xi_{(t_1', t_1' +W)} (\sigma_{t_1}(\chi_1)) \neq {\vec \phi}^m$ and for all $t_2' \in [0, V_2)$, $\Xi_{(t_2', t_2' +W)} (\sigma_{t_2}(\chi_2)) \neq {\vec \phi}^m$.

\end{proof}

\bibliographystyle{spbasic}
\bibliography{acyclic_complexity}

\end{document}